\documentclass [11pt]{article}
\usepackage{color}
\usepackage{amsmath}
\usepackage{amssymb}
\usepackage{amscd}
\usepackage{amsthm}
\usepackage{amsopn}
\usepackage{verbatim}
\usepackage{amsmath}
\usepackage{amsfonts}
\usepackage[active]{srcltx}
\usepackage{empheq}
\usepackage{graphicx}
\definecolor{SeaGreen}{RGB}{46,139,87}
\definecolor{Navy}{RGB}{0,0,128}
\definecolor{Maroon}{RGB}{128,0,0}
\newtheorem{theorem}{Theorem}
\newtheorem{lemma}{Lemma}[section]
\newtheorem{proposition}{Proposition}

\newtheorem*{acknowledgment*}{Acknowledgment}
\newtheorem{remark}{Remark}[section]

\newcommand{\R}{\mathbb R}

\newcommand{\A}{\mathcal A}
\newcommand{\B}{\mathcal B}

\DeclareMathOperator{\Div}{div}
\DeclareMathOperator{\curl}{curl}

\def\XXint#1#2#3{{\setbox0=\hbox{$#1{#2#3}{\int}$ }
\vcenter{\hbox{$#2#3$ }}\kern-.6\wd0}}

\def\Hg {{\mathcal H}}

\def\Jg {{\mathcal J}}

\newcommand{\LL}{\mathcal L}

\newcommand{\OO}{\mathcal O}

\def\Wg {{\mathcal W}}
\def\Vg {{\mathcal V}}

\def\gf {{\mathfrak g}}

\begin{document}
\bibliographystyle{ieeetr} \date{} \title{Existence of superconducting
  solutions for a reduced {G}inzburg-{L}andau model in the presence of
  strong electric currents.}  \author{Yaniv Almog\thanks{Department of
    Mathematics, Louisiana State University, Baton Rouge, LA 70803,
    USA}\,\,\thanks{Present address: Department of Mathematics, Ort Braude
    College, Carmiel 21610, Israel}, Leonid
  Berlyand\thanks{Department of Mathematics, Pennsylvania State
    University, University Park, PA 16802, USA}, Dmitry
  Golovaty\thanks{Department of Mathematics, The University of Akron,
    Akron, Ohio 44325, USA}, and Itai Shafrir\thanks{Department of
    Mathematics, Technion - Israel Institute of Technology, 32000
    Haifa, Israel}}

\maketitle
\begin{abstract}
  In this work we consider a reduced Ginzburg-Landau model in which
  the magnetic field is neglected and the magnitude of the current density
  is significantly stronger than that considered in a recent work by
  the same authors. We prove the existence of a solution which can be obtained by solving a non-convex minimization
  problem away from the boundary of the domain. Near the boundary, we show that this solution is
  essentially one-dimensional.
\end{abstract}
\section{Introduction}
\label{sec:2}
It is a well-established fact that superconductors are characterized by (i) a complete loss of electrical 
resistivity below a certain critical temperature and (ii) a complete expulsion of the magnetic field from purely superconducting
regions. It is thus possible, in a purely superconducting state,
to run electric current through a superconducting
sample while generating only a vanishingly small voltage
drop.  However, a sufficiently strong
electric current will destroy superconductivity and revert the material to the normal
state---even while it remains below the critical temperature \cite{alhe14}.

In this work, we study the effect of electric current on
superconductivity through the 
time-dependent Ginzburg-Landau model \cite{chhe03,goel68} presented
here in a dimensionless form
   \begin{subequations}
\label{eq:1}
\begin{empheq}[left={\empheqlbrace}]{alignat=2}
  &\frac{\partial u}{\partial t} +   i\phi u = \left(\nabla - iA \right)^{2} u + u 
  \left( 1 - |u|^{2} \right) & \text{ in } & \Omega\times\R_+ \,,\\
\; & \kappa^2\left(\curl{A}\right)^2 + \sigma \left(\frac{\partial A}{\partial t} + 
 \nabla\phi\right)  =   \Im\{\bar{u} \nabla u\}  + |u|^{2}A &
\text{ in } & \Omega\times\R_+\,, \\
 &(i\nabla+A)u\cdot{\bf n}=0 \quad\text{and}\quad -\sigma\left(\frac{\partial A}{\partial t} + 
 \nabla\phi\right) \cdot{\bf n}=J& \text{ on } & \partial\Omega\times\R_+ \,, \\
&u(x_,0)=u_0\quad\text{and}\quad  A(x_,0)=A_0& \text{ in } & \Omega\,, \\
\end{empheq}
\end{subequations}
In the above system of equations, $u$ is the order parameter with
$|u|^2$ representing the number density of superconducting electrons.
When $|u| = 1$, a material is said to be purely superconducting
while it is in the normal state when $u = 0$.  We
denote the magnetic vector potential by $A$---so that the magnetic
field is given by $h=\curl A=\partial_1A_2-\partial_2A_1$---and by $\phi$ the
electric scalar potential. The constants $\kappa$ and $\sigma$ are the
Ginzburg-Landau parameter and normal conductivity, of the
superconducting material, respectively, and the quantity
$-\sigma(A_t+\nabla\phi)$ is the normal current.  All lengths in \eqref{eq:1}
have been scaled with respect to the coherence length $\xi$ that
characterizes spatial variations in $u$.  The domain $\Omega\subset\R^2$ is assumed to be smooth (at least $C^{3,\alpha}$ for some $\alpha>0$), the outward unit normal
vector to $\partial\Omega$ is denoted by ${\bf n}$, and the function 
$J:\partial\Omega\to\R$ represents the normal current entering the sample. Here the boundary current must satisfy $\int_{\partial\Omega}J=0$. Note, that it is possible to
prescribe the electric potential on $\partial\Omega$ instead of the current.

It has been demonstrated in \cite{al12} that, when both the current $J$ and the domain $\Omega$ are fixed, one can formally obtain from \eqref{eq:1} the
following system of equations
\begin{subequations}
\label{eq:2}
\begin{empheq}[left={\empheqlbrace}]{alignat=2}
&  \frac{\partial u}{\partial t} + i\phi u = \Delta u + u\left( 1 - | u|^{2} \right), & \qquad &\text{in } \Omega\times\R_+, \\
&  \sigma\Delta\phi = \nabla\cdot [\Im(\bar{ u} \nabla u)],
& \qquad & \text{in }  \Omega\times\R_+, \\
&  \frac{\partial u}{\partial{\bf n}}=0 \text{ and }-\sigma\frac{\partial\phi}{\partial{\bf n}}=J, &\qquad &\text{on } \partial\Omega\times\R_+,  \\
&  u(x_,0)= u_0, & \qquad &\text{in } \Omega \,,
\end{empheq}
\end{subequations}
when $\kappa\to\infty$. This limit corresponds to the regime when the domain size is much smaller than the 
penetration depth  $\lambda=\kappa\xi$ which characterizes variations in the
magnetic field. Similarly to \eqref{eq:1}, the system \eqref{eq:2} remains invariant
under the transformation 
\begin{displaymath}
u\to e^{i\omega(t)}u \quad ; \quad \phi\to\phi-\frac{\partial\omega}{\partial t} \,.
\end{displaymath}
In what follows, we set
\begin{displaymath}
  \omega=\int_0^t\frac{(|u|^2\phi)_\Omega(\tau)}{(|u|^2)_\Omega(\tau)} \,d\tau \,,
\end{displaymath}
to guarantee that
\begin{equation}
\label{eq:3}
  (|u|^2\phi)_\Omega(t)\equiv0
\end{equation}
for all $t>0$. Here
\begin{equation}
\label{eq:ave}
(f)_\Omega=\int_\Omega f
\end{equation}
is the average of $f:\Omega\to\mathbb R$.

In the present contribution we consider steady-state solutions of
\eqref{eq:2} in the large domain limit. Let then $\Omega_\epsilon$ be obtained
from $\Omega$ via the map
\begin{displaymath}
  x\to\frac{x}{\epsilon} \,,
\end{displaymath}
where $\epsilon$ is a small parameter. 
Let $(u,\phi)$ denote a smooth stationary solution of
\eqref{eq:2} in $\Omega_\epsilon$, which must therefore satisfy 
\begin{subequations}
\label{eq:4}
\begin{empheq}[left={\empheqlbrace}]{alignat=2}
-&\Delta u + i\phi u = u\left( 1 - | u|^{2} \right), & \text{in
} \Omega_\epsilon \,, \\
&  \sigma\Delta\phi = \nabla\cdot [\Im(\bar{ u} \nabla u)], &  \text{in }  \Omega_\epsilon\,, \\
&  \frac{\partial u}{\partial{\bf n}}=0,\ -\sigma\frac{\partial\phi}{\partial{\bf n}}=J, \text{ and
}\int_{\partial\Omega}J=0, & \text{\ \ \ \ on } \partial\Omega_\epsilon\,.
\end{empheq}
\end{subequations}
Applying the transformation 
\begin{equation}
\label{eq:5}
 x^1=\epsilon x , \;\;  J^1= \frac{J}{\epsilon}  ,\;\; \phi^1=\frac{\phi}{\epsilon^2} 
  ,\;\; \sigma^1=\sigma\epsilon^2 \,,
\end{equation}
to \eqref{eq:4} and dropping the superscript for notational convenience yields 
\begin{subequations}
\label{eq:6}
\begin{empheq}[left={\empheqlbrace}]{alignat=2}
-&\Delta u + i\phi u = \frac{u}{\epsilon^2}\left( 1 - | u|^{2} \right), & \text{in } \Omega\,, \\
&  \sigma\Delta\phi = \nabla\cdot [\Im(\bar{ u} \nabla u)], &  \text{in }  \Omega\,, \\
&  \frac{\partial u}{\partial{\bf n}}=0,\ -\sigma\frac{\partial\phi}{\partial{\bf n}}=J, \text{ and }\int_{\partial\Omega}J=0, & \text{\ \ \ \ on } \partial\Omega\,.
\end{empheq}
\end{subequations}
In view of \eqref{eq:5} we set \[\sigma=\sigma_0\epsilon^2\] in the sequel, but we will continue
using $\sigma$ to simplify the notation whenever there is no ambiguity. 

Setting $u=\rho e^{i\chi}$, leads to the following problem
\begin{subequations}
\label{eq:7}
\begin{empheq}[left={\empheqlbrace}]{alignat=2}
-&\Delta\rho+\rho|\nabla\chi|^2 = \frac{\rho}{\epsilon^2}(1 - \rho^2), & \text{in } \Omega\,, \\
&\Div(\rho^2\nabla\chi)=\rho^2\phi, & \text{in } \Omega\,, \\
&  \sigma\Delta\phi = \Div(\rho^2\nabla\chi),  &  \text{in }  \Omega\,, \\
&  \frac{\partial\rho}{\partial{\bf n}}=\frac{\partial\chi}{\partial{\bf n}}=0,\ -\sigma\frac{\partial\phi}{\partial{\bf n}}=J,\text{ and }\int_{\partial\Omega}J=0, & \text{\ \ \ \ on } \partial\Omega\,.
\end{empheq}
\end{subequations}
Note that \eqref{eq:7} is invariant with respect to the transformation
$\chi\to\chi+C$ for any constant $C$. We set $(\chi)_\Omega=0$ in order to eliminate
this degree of freedom throughout the paper.

The systems \eqref{eq:2}, \eqref{eq:6}, and \eqref{eq:7} have
attracted significant interest among both physicists and
mathematicians with 
\cite{al12}--\nocite{ivko84,ivetal82,doetal98,voetal03,kabe14,ruetal07,rust08,ruetal10}\cite{ruetal10b} addressing a variety of related problems in a one-dimensional setting for a variety boundary
conditions and \cite{al08} that considers the linearized version of
\eqref{eq:2} in higher dimensions.  A different simplification of
\eqref{eq:1} was derived in \cite{dugr96} for the same limit
$\kappa\to\infty$ under an additional assumption that $J$ and $\sigma$ are of order $\OO(\kappa^2)$ (cf.
\cite{duetal10}). In \cite{aletal15} it has been shown that, when
$\|J\|_{H^{3/2}(\partial\Omega)}\ll 1/\epsilon$, the system \eqref{eq:6} possesses a stationary solution
$(u_s,\phi_s)$ satisfying
  \begin{displaymath}
    \|1-|u_s|\,\|_{2,2} \leq C\epsilon^2\|J\|_{H^{3/2}(\partial\Omega)}^2 \,.
  \end{displaymath}
Furthermore, it is established in \cite{aletal15} that $(u_s,\phi_s)$ is
linearly stable. 

In the present paper we focus on steady state solutions of
\eqref{eq:6} when $\|J\|_\infty\sim\OO(1/\epsilon)$ (or, equivalently, solutions of
\eqref{eq:4} for $\|J\|_\infty\sim\OO(1)$). Far away from the boundaries we
approximate the solution $(u_s,\phi_s)$ by the following ``outer" ansatz 
\[\phi_s\approx0,\qquad |u_s|^2\approx1-|\nabla\zeta|^2,\qquad \chi_s\approx\zeta/\epsilon,\]
where $\zeta$ solves 
\begin{equation}
\label{eq:8}
  \begin{cases}
  \Div \left([1-|\nabla\zeta|^2]\nabla\zeta \right) =0, & \text{in } \Omega, \\
  [1-|\nabla\zeta|^2]\,\frac{\partial\zeta}{\partial{\bf n}} = j\mbox{ and }\int_{\partial\Omega}j=0, & \text{on } \partial\Omega\,,
  \end{cases}
\end{equation}
and $j=\epsilon J$. 

The following existence/uniqueness result holds for the system \eqref{eq:8}. 
\begin{proposition}
\label{prop:exist}
For any $0<\alpha<1$ and $k\in\mathbb N$, let
\begin{displaymath}
    \Wg_{k,\alpha} := \left\{ u\in C^{k,\alpha}(\bar{\Omega})\,\Big| \int_\Omega u\,dx=0\,\right\}.
  \end{displaymath}
Suppose that $j=\mu j_r$ with $j_r\in
C^{k,\alpha}(\partial\Omega)$ satisfying $\int_{\partial\Omega}j_r=0$ and $\|j_r\|_\infty =1$. Then, for a sufficiently small $\mu>0$ and any $k\geq3$, there exists a solution
$\zeta_\mu\in\Wg_{k,\alpha}$ of \eqref{eq:8}. If in addition
\begin{equation}
\label{eq:9}
\|\nabla\zeta_\mu\|_\infty< \frac{1}{\sqrt{3}} \,,
\end{equation}
then
\begin{displaymath}
  \|\nabla\zeta_\mu\|_\infty=\|\nabla\zeta_\mu\|_{L^\infty(\partial\Omega)}
\end{displaymath}
and there exists at most one solution of \eqref{eq:8} satisfying
\eqref{eq:9}. 
\end{proposition}

Once the existence of the ``outer'' solution has been demonstrated, we
establish the existence of a boundary layer approximation. Suppose that $\partial\Omega$ is parametrized by $s$ and let
$t=d(x,\partial\Omega)$, where $d(x,\partial\Omega)$ is the distance from $x$ to $\partial\Omega$. We set $\tau=t/\epsilon$ and define \[\rho_i(s,\tau)=\left|u\left(x(s,\tau)\right)\right|,\qquad\varphi_i(s,\tau)=\epsilon^2\phi\left(x(s,\tau)\right).\]  Let $(\rho_{i0},\varphi_{i0})$
denote the formal limit of $(\rho_i,\varphi_i)$ as $\epsilon\to0$ near the boundary. 
The leading order boundary layer term is given by the solution of
\begin{equation}
\label{eq:10}
  \begin{cases}
-\rho^{\prime\prime}_{i0} - \Big(\rho_r^2-\frac{(\sigma_0\varphi_{i0}^\prime-j)^2}{\rho_{i0}^4}-\rho_{i0}^2\Big)\rho_{i0} = 0 & \text{in } \R_+, \\
   -\sigma_0\varphi_{i0}^{\prime\prime}+ \rho_{i0}^2\varphi_{i0} = 0 & \text{in } \R_+, \\
    \rho_{i0}^\prime(0)=0, & \\
     \varphi_{i0}^\prime(0)= \frac{j}{\sigma_0}\,,
  \end{cases}
\end{equation}
where the derivatives are taken in $\tau$ and $\rho_r^2(s) = 1-\left|\partial\zeta/\partial s(s,0)\right|^2$. We prove the following lemma.
\begin{proposition}
\label{lem:2.2}
Suppose that (\ref{eq:9}) holds. For a sufficiently large $\sigma_0$,
there exists a solution $(\rho_{i0},\varphi_{i0})$ of \eqref{eq:10} such
that $\rho_{i0}\geq\left(1-\left|\nabla\zeta(\cdot,0)\right|^2\right)^{1/2}$ and $(\rho_{i0}-1+|\nabla\zeta(\cdot,0)|^2,\varphi_{i0})\in
H^1(\R_+,\R^2)$.  Furthermore, for some positive $\gamma$ and $C$, we have that
  \begin{equation}
\label{eq:38}
  |\varphi_{i0}|+|\varphi_{i0}^\prime|+
  \big|\rho_{i0}^2-1+|\nabla\zeta(\cdot,0)|^2\big|^{1/2}+|\rho_{i0}^\prime| <Ce^{-\gamma \sigma_0^{-1/2}\tau}. \\
    \end{equation}
\end{proposition}
We also demonstrate that, when $\sigma_0\to\infty$, a good approximation for
$(\rho_{i0},\varphi_{i0})$ is given by the solution of
\begin{displaymath}
  \begin{cases}
\rho_r^2-\frac{(\sigma_0\varphi_{i0}^\prime-j)^2}{\rho_{i0}^4}-\rho_{i0}^2 = 0 & \text{in } \R_+, \\
   -\sigma_0\varphi_{i0}^{\prime\prime}+ \rho_{i0}^2\varphi_{i0} = 0 & \text{in } \R_+, \\
    \rho_{i0}^\prime(0)=0, & \\
     \varphi_{i0}^\prime(0)= \frac{j}{\sigma_0}\,. 
  \end{cases}
\end{displaymath}

Finally, we combine the solutions of (\ref{eq:8}) and (\ref{eq:10}) to
obtain a uniform approximation of $(\rho,\chi,\phi)$ that we denote by
$(\rho_0,\chi_0,\phi_0)$ (cf. \eqref{eq:97}). This approximation satisfies
$\rho_0\approx 
\left(1-|\nabla\zeta|^2\right)^{1/2}$ for $d(x,\partial\Omega)\gg\epsilon$ and $\rho_0\approx\rho_{i0}$ when
$d(x,\partial\Omega)\sim\OO(\epsilon)$. We establish
existence of a solution for the system (\ref{eq:7}) in the following
\begin{theorem}
\label{thm:stationary2}
Suppose that for a given $j\in C^{3,\alpha}(\partial\Omega)$ the problem \eqref{eq:8} has
a solution satisfying $\|\nabla\zeta\|_\infty^2<5/14-\sqrt{65}/70$. Then, for 
sufficiently large $\sigma_0$ and  $s<1$ there exist positive $\epsilon_0(\Omega,j,\sigma_0,s)$ and
$C(\Omega,j,\sigma_0,s)$, such that (\ref{eq:7}) possesses a solution satisfying
 \begin{equation}
\label{eq:11}
\|\rho-\rho_0\|_2 + \|\chi-\chi_0\|_{1,2}
+ \|\phi-\phi_0\|_2  \leq
C\epsilon^{2+s/2} \,,
  \end{equation}
and
\begin{equation}
\label{eq:12}
 \|\rho-\rho_0\|_{1,2} + \|\chi-\chi_0\|_{2,2}
+ \|\phi-\phi_0\|_{1,2}+  \epsilon( \|\rho-\rho_0\|_{2,2} +  \|\phi-\phi_0\|_{2,2}) \leq
 C\epsilon^{1+s/2},
\end{equation}
for all $\epsilon<\epsilon_0$.
\end{theorem}

Note that Theorem \ref{thm:stationary2} holds even when
$\|j\|_{C^{3,\alpha}(\partial\Omega)}$ is not necessarily small, as long as there exists a solution for
\eqref{eq:8} which satisfies  $\|\nabla\zeta\|_\infty^2<5/14-\sqrt{65}/70$. However, because $5/14-\sqrt{65}/70<1/3$, Theorem
\ref{thm:stationary2} is not optimal. We elaborate on this point further
in Section 4.

The two-dimensional analysis reveals a much richer and
more complex picture compared to what is observed for the one-dimensional problem. The
physical relevance of the present work is manifested primarily through the study of the "outer" problem \eqref{eq:8}. When a solution to \eqref{eq:8} exists, we can
conclude from Theorem \ref{thm:stationary2} that, for a sufficiently
large $\sigma_0$, a purely superconducting solution exists for \eqref{eq:4}
as well. On the other hand, when the current density on the boundary is
such that \eqref{eq:8} has no solution, it is reasonable to
expect that a purely superconducting solution cannot exist for
\eqref{eq:4} either. If this is indeed the case, then the question of existence of solution for \eqref{eq:4}
is reduced to the much simpler problem \eqref{eq:8}. We note, however,
that typically, superconducting samples are much larger than
the penetration depth. The magnetic field induced by the electric
current would then destroy superconductivity for much weaker current
densities than the critical values we find here (cf. for instance
\cite{aletal17}).

The rest of the paper is organized as follows. In the next section we
consider an outer approximation of solutions of \eqref{eq:7} and in
particular, prove Proposition \ref{prop:exist} and obtain a higher
order term which is necessary in the proof of Theorem
\ref{thm:stationary2}. In Section 3, we consider a matching inner
solution, while Section 4 is devoted to the development of a uniform
approximation and the proof of Theorem \ref{thm:stationary2}.

\section{Outer approximation}
\label{sec:4}

Suppose that $\|J\|_{C^{3,\alpha}(\partial\Omega)}\epsilon\sim\OO(1)$ for some $\alpha>0$ and
$\sigma=\sigma_0\epsilon^2$ for some $\sigma_0>0$. The solution of the problem
\eqref{eq:7} will be obtained by "gluing" an outer and an inner
approximations, valid away from and close to the boundary,
respectively. In this section we construct the outer approximation.

Combining (\ref{eq:7}b-c), we note that $\phi$ satisfies
\[\sigma_0\epsilon^2\Delta\phi=\rho^2\phi\,.\] We assume that the outer solution---outside
of a $\OO(\epsilon)$-thin inner layer near $\partial\Omega$---corresponds to a
superconducting state in which the magnitude $\rho$ of the order
parameter is bounded away from zero. By following the standard
argument, this implies that $\phi$ is exponentially small in the outer
region so that we can set $\phi_{out}\equiv0$ away from the boundary. This
observation leads to the approximate problem
\begin{equation}
  \label{eq:13}
  \begin{cases}
-\Delta\rho_{out}+\rho_{out}|\nabla\chi_{out}|^2 = \frac{\rho_{out}}{\epsilon^2}\left(1 - \rho_{out}^2\right), & \text{in } \Omega\,, \\
\Div(\rho_{out}^2\nabla\chi_{out})=0, &  \text{in }  \Omega\,, \\
\frac{\partial\rho_{out}}{\partial{\bf n}}=0, \quad \rho_{out}^2\frac{\partial\chi_{out}}{\partial{\bf n}}=-J,\text{ and }\int_{\partial\Omega}J=0, & \text{on }
\partial\Omega \,.
  \end{cases}
\end{equation}
Here the boundary condition reflects the conjecture that the normal
current $-\sigma\nabla\phi$ turns into a superconducting one
$\rho_{out}^2\nabla\chi_{out}$ within a thin one-dimensional boundary layer.
We seek an approximation to the solution of \eqref{eq:13} in the limit
$\epsilon\to0$. 

\subsection{Solution of the outer problem---$\mathcal{O}(1)$-term}
Since $J\sim\OO\left(\frac{1}{\epsilon}\right)$ and
$\rho_{out}\sim\OO(1)$ is bounded away from $0$, we have that
$\nabla\chi_{out}\sim\OO(1/\epsilon).$ Using the first equation in \eqref{eq:13}, we
obtain that to leading order
\begin{equation}
\label{eq:14}
  |\nabla\chi_{out,0}|^2 = \frac{1}{\epsilon^2}\left(1 - \rho_{out,0}^2\right)\,,
\end{equation}
hence
\begin{equation}
\label{eq:6.1}
  \begin{cases}
  \Div \left([1-\epsilon^2|\nabla\chi_{out,0}|^2]\nabla\chi_{out,0} \right) =0, & \text{in } \Omega, \\
  [1-\epsilon^2|\nabla\chi_{out,0}|^2]\,\frac{\partial\chi_{out,0}}{\partial{\bf n}} = -J\text{ and }\int_{\partial\Omega}J=0, & \text{on } \partial\Omega\,.
  \end{cases}
\end{equation}
Rescaling
\begin{displaymath}
  \zeta= \epsilon\chi_{out,0},\ \ \  \quad j=-\epsilon J,
\end{displaymath}
yields
  \begin{equation}
  \label{eq:15}
  \begin{cases}
  \Div \left([1-|\nabla\zeta|^2]\nabla\zeta \right) =0, & \text{in } \Omega, \\
  [1-|\nabla\zeta|^2]\,\frac{\partial\zeta}{\partial{\bf n}} = j\text{ and }\int_{\partial\Omega}j=0, & \text{on } \partial\Omega\,.
  \end{cases}
\end{equation}
We seek smooth (at least, $H^3(\Omega)$) solutions of \eqref{eq:15} satisfying 
\begin{equation}
\label{icon}
\|\nabla\zeta\|_\infty\leq 1.
\end{equation}
 The latter inequality is needed to guarantee that $\rho_{out,0}$ in \eqref{eq:14} remains meaningful. In fact, as will be discussed below, a stronger bound on $\|\nabla\zeta\|_\infty$ will be required to establish existence of solutions of \eqref{eq:15}.

First, we observe that the inequality constraint \eqref{icon} and the elementary statement  
\[
\max_{t\in[0,1]}\left(t-t^3\right)=\frac{2}{3\sqrt{3}},
\]
result in the bound 
\begin{equation}
\label{scur}
\|\mathbf{j}_s\|_{\infty}\leq\frac{2}{3\sqrt{3}}\,,
\end{equation}
on the superconducting current $$\mathbf{j}_s:=\left[1-|\nabla\zeta|^2\right]\nabla\zeta.$$ A necessary condition for existence of solutions of \eqref{eq:15} then follows from the requirement that the boundary data should satisfy \eqref{scur}, that is $\|j\|_{L^\infty(\partial\Omega)}\leq\frac{2}{3\sqrt{3}}$. In fact, even stronger necessary condition can be established as we will demonstrate in the next proposition. 

Let the
distance between any two points $x,y\in\Omega$ be defined as
\begin{displaymath}
  d(x,y) = \inf_{\gamma\subset\bar{\Omega}}\mathbf{L}(\gamma) \,,
\end{displaymath}
where $\gamma$ is a continuous path connecting $x$ and $y$ and $\mathbf{L}(\gamma)$ is the length of $\gamma$. If $\Omega$ is
convex, then $d(x,y)=|x-y|$. For any $x,y\in\partial\Omega$, let
\begin{displaymath}
  M(x,y) = \frac{1}{d(x,y)}\left|\int_{\Gamma}j\,ds\right|,
\end{displaymath}
where $\Gamma$ is either of the two possible paths in $\partial\Omega$ connecting $x$ and $y$ (note that the value of $M$ is independent of the choice of the path due to the condition $\int_{\partial\Omega}j=0$). 
\begin{proposition}
  \label{prop:nonexist}
Suppose that the boundary value problem \eqref{eq:15} has a solution in $H^3(\Omega)$. Then
\begin{equation}
\label{eq:16}
\sup_{x,y\in\partial\Omega} M(x,y)
\leq\frac{2}{3\sqrt{3}}.
\end{equation}
\end{proposition}
\begin{proof}
  Given $x,y\in\partial\Omega$, let $\tilde\Gamma$ denote a shortest path in $\bar\Omega$
  connecting $y$ to $x$.  Suppose that a solution $\zeta\in H^3(\Omega)$ of
  \eqref{eq:15} exists. Let $\Omega_\Gamma\subset\Omega$ denote the domain enclosed in
  $\Gamma\cup\tilde{\Gamma}$ and ${\bf n}$ denote the outward unit normal on
  $\partial\Omega_\Gamma$.   Clearly,
\begin{equation}
\label{eq:17}
  \int_{\Gamma}j\,ds + \int_{\tilde\Gamma}\left[1-|\nabla\zeta|^2\right]\frac{\partial\zeta}{\partial{\bf n}} \,ds = 0
\end{equation}
and the inequality 
\begin{displaymath}
   \left|\int_{\tilde\Gamma}\left[1-|\nabla\zeta|^2\right]\frac{\partial\zeta}{\partial{\bf n}} \,ds \right| \leq
   \|\mathbf{j}_s\|_\infty d(x,y) \leq
   \frac{2}{3\sqrt{3}} d(x,y)
\end{displaymath}
holds by \eqref{scur}. Substituing this expression into \eqref{eq:17} yields
\begin{displaymath}
  M(x,y)\leq \frac{2}{3\sqrt{3}}
\end{displaymath}
and \eqref{eq:16} follows because $x,y\in\partial\Omega$ were chosen arbitrarily.
\end{proof} 

\begin{figure}[htb]
  \centering
 \includegraphics[height=1.5in]{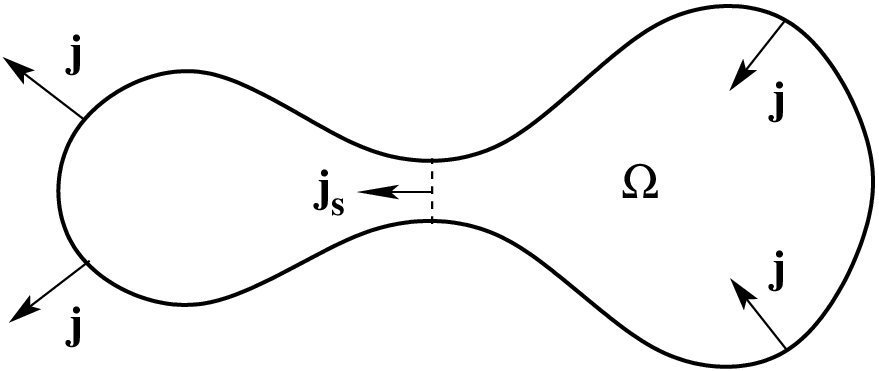}
\caption{Global constraint on the boundary data---the superconducting current through the narrowing in $\Omega$ has to satisfy the pointwise bound $|\mathbf{j}_s|\leq\frac{2}{3\sqrt{3}}$.}
  \label{fig:1}
\end{figure}

\begin{remark}
{\em It can be easily verified that a solution $\zeta= jx$ of \eqref{eq:15} obtained in \cite{ivko84} when $\Omega=\mathbb{R}$, exists 
only if the current at infinity satisfies $|j|\leq2/(3\sqrt{3})$. Proposition \ref{prop:nonexist} establishes that the geometry of the domain $\Omega$ places additional restrictions on the boundary data in \eqref{eq:15} when $\Omega\subset\mathbb{R}^2$, as illustrated in Figure \ref{fig:1}.}
\end{remark}

In order to prove existence of solutions to \eqref{eq:15} we will need the following auxiliary lemma
\begin{lemma}
  \label{lem:max}
Let $\zeta$ denote a $C^3(\Omega)$ solution of \eqref{eq:15}. If
$\|\nabla\zeta\|_\infty<1/\sqrt{3}$, then $\|\nabla\zeta\|_\infty = \|\nabla\zeta\|_{L^\infty(\partial\Omega)}$.
\end{lemma}
\begin{proof}
By expanding the left hand side of \eqref{eq:15} we obtain
\[\left(1-|\nabla\zeta|^2\right)\Delta\zeta-\nabla\left(|\nabla\zeta|^2\right)\cdot\nabla\zeta=0,\]
then taking the gradient of both sides in this expression, in turn, gives
\[\left(1-|\nabla\zeta|^2\right)\nabla\Delta\zeta-\Delta\zeta\nabla\left(|\nabla\zeta|^2\right)-\nabla\nabla\left(|\nabla\zeta|^2\right)\nabla\zeta-\nabla\nabla\zeta\nabla\left(|\nabla\zeta|^2\right)=0.\]
Solving for $\nabla\Delta\zeta$ and taking dot product with $\nabla\zeta$, we have
\begin{multline}
\label{greq}
  \nabla\zeta\cdot\nabla\Delta\zeta = \frac{1}{\left(1-|\nabla\zeta|^2\right)}\left(\Delta\zeta \,\nabla\zeta\cdot\nabla(|\nabla\zeta|^2) \right. \\ \left. +
  \frac{1}{2}|\nabla(|\nabla\zeta|^2)|^2 + \nabla\zeta\cdot \nabla\nabla(|\nabla\zeta|^2)\nabla\zeta\right).
\end{multline}
Now let
\begin{displaymath}
  \LL  = A(\nabla\zeta):\nabla\nabla - b\cdot\nabla \,, 
\end{displaymath}
where $A:\mathbb{R}^2\to M^{2\times2}$  and $b\in C^1(\Omega,\R^2)$ are given by
\begin{equation}
\label{eq:18}
  A(z)= \left(1-|z|^2\right)\mathrm{I} - 2\,z\otimes z
\end{equation}
and 
\begin{displaymath}
  b= \nabla(|\nabla\zeta|^2) + 2\,\Delta\zeta\,\nabla\zeta\,, 
\end{displaymath}
respectively. Observing that
  \begin{equation}
  \label{eqlap}
    \Delta(|\nabla u|^2)= 2 |\nabla\nabla u|^2 + 2\nabla u\cdot\nabla\Delta u\,,
  \end{equation}
holds for any $u\in C^3(\Omega)$, substituting \eqref{greq} into \eqref{eqlap}, and rearranging terms we find that
\begin{displaymath}
  \LL |\nabla\zeta|^2 = 2 \left(1-|\nabla\zeta|^2\right)|\nabla\zeta|^2 \geq 0 \,.
\end{displaymath}
The operator $\LL$ is uniformly elliptic when
$|\nabla\zeta|<1/\sqrt{3}$. Indeed, if we let $\nabla\zeta^\perp:=\left(-\zeta_y,\zeta_x\right),$ it
immediately follows from \eqref{eq:18} that $\nabla\zeta^\perp$ and $\nabla\zeta$ are
eigenvectors of $A(\nabla\zeta)$ with eigenvalues $1-|\nabla\zeta|^2$ and $1-3|\nabla\zeta|^2$,
respectively. Both eigenvalues are positive and the matrix $A(\nabla\zeta)$ is
positive definite as long as $|\nabla\zeta|<1/\sqrt{3}$. The proof of the lemma
follows by the maximum principle.  
\end{proof}

We now prove the first existence and uniqueness result for \eqref{eq:15}. For convenience, we repeat here the statement of Proposition \ref{prop:exist}.

\medskip\noindent
{\bf Proposition 1.} {\em 
For any $0<\alpha<1$ and $k\in\mathbb N$, let
\begin{displaymath}
    \Wg_{k,\alpha} := \left\{ u\in C^{k,\alpha}(\bar{\Omega})\,\Big| \int_\Omega u\,dx=0\,\right\}.
  \end{displaymath}
Suppose that $j=\mu j_r$ with $j_r\in
C^{k,\alpha}(\partial\Omega)$ satisfying $\int_{\partial\Omega}j_r=0$ and $\|j_r\|_\infty =1$. Then, for a sufficiently small $\mu>0$ and any $k\geq3$, there exists a solution
$\zeta_\mu\in\Wg_{k,\alpha}$ of \eqref{eq:15}. If in addition
\begin{equation}
\label{eq:9.1}
\|\nabla\zeta_\mu\|_\infty< \frac{1}{\sqrt{3}} \,,
\end{equation}
then
\begin{equation}
\label{eq:9.2}
  \|\nabla\zeta_\mu\|_\infty=\|\nabla\zeta_\mu\|_{L^\infty(\partial\Omega)}
\end{equation}
and there exists at most one solution of \eqref{eq:15} satisfying
\eqref{eq:9.1}.}

\smallskip
\begin{proof} 
We use the implicit function theorem. 
Define $F:\Wg_{k,\alpha}\times\R \to Z$ by
\begin{displaymath}
  F(u,\mu) = \left(-\Div\left(\left[1-|\nabla u|^2\right]\nabla u\right),\left[1-|\nabla u|^2\right]\nabla u\cdot{\bf n}-\mu j_r\right)\,,
\end{displaymath}
where
\begin{displaymath}
  Z=\{(z_1,z_2)\in C^{k-2,\alpha}(\Omega)\times C^{k-1,\alpha}(\partial\Omega)\,;\, \int_\Omega z_1+\int_{\partial\Omega}z_2=0\}\,.
\end{displaymath}

Clearly, $F(0,0)=0$. Furthermore, it can be readily verified that the
Frechet derivative of $F$ with respect to $u$ is  given by
\begin{displaymath}
  D_uF(u,\mu)\omega =  \left(-\Div\left(A(\nabla u)\nabla \omega\right),A(\nabla u)\nabla \omega\cdot{\bf n}\right)\,,
\end{displaymath}
where $A$ is as defined in \eqref{eq:18}. In particular, by standard
elliptic estimates,
\begin{displaymath}
  D_uF(0,0)\omega =\left(-\Delta \omega,\frac{\partial \omega}{\partial{\bf n}}\right) \,,
\end{displaymath}
is an isomorphism of $\Wg_{k,\alpha}$ onto $Z$. It follows that a solution of \eqref{eq:15} with $j=\mu j_r$
exists in some neighborhood of $\mu=0$. We denote this solution by
$\zeta_\mu$. 

By the implicit function theorem, the solution continues to exist
 as long as $D_uF(\zeta_\mu,\mu)$ is an isomorphism from $\Wg_{k,\alpha}$ onto
 $Z$. To show that this is indeed the case---as long as \eqref{eq:9.1}
 holds---we first notice that the map $D_uF(\zeta_\mu,\mu)$ is
 injective. Indeed,  if 
 $D_uF(\zeta_\mu,\mu)\omega=0$ for some $\omega\in\Wg_{k,\alpha}$, i.e.,
  \begin{equation}
    \label{eq:20}
 \left\{
 \begin{aligned}
 -\Div\left(A(\nabla\zeta_\mu)\nabla \omega\right)\omega=0 ~\text{ in }\Omega\,,\\
  A(\nabla\zeta_\mu)\nabla \omega\cdot{\bf n}=0 ~\text{ on }\partial\Omega\,,
   \end{aligned} 
\right.
 \end{equation}
    then multiplying the equation in \eqref{eq:20} by $\omega$ and
    using Green formula yields, 
\begin{displaymath}
  0= \int_\Omega A\left(\nabla \zeta_\mu\right)\nabla\omega\cdot\nabla\omega \,dx \geq \|1-3|\nabla\zeta_\mu|^2\|_\infty \|\nabla \omega\|_2\,,
\end{displaymath}
 implying that $\omega=0$ (under the assumption \eqref{eq:9.1}).
 To prove that $D_uF(\zeta_\mu,\mu)$ is onto $Z$, consider any $(z_1,z_2)\in Z$
 and the problem,
\begin{equation}
    \label{eq:21}
 \left\{
 \begin{aligned}
 -\Div\left(A(\nabla\zeta_\mu)\nabla \omega\right)\omega=z_1 ~\text{ in }\Omega\,,\\
  A(\nabla\zeta_\mu)\nabla \omega\cdot{\bf n}=z_2 ~\text{ on }\partial\Omega\,.
   \end{aligned} 
\right.
 \end{equation}
 Existence for \eqref{eq:21} can be established by using the Schauder
 approach along with the Fredholm alternative (see Theorem 6.30, Theorem 6.31 and the subsequent remark in
 \cite{gitr01}). Alternatively, one can
 prove existence of a weak solution $\omega\in H^1(\Omega)$ to \eqref{eq:21} by
 minimization of the functional
 \begin{displaymath}
   J(u)=\int_\Omega \Big[A\left(\nabla \zeta_\mu\right)\nabla u\cdot\nabla u -z_1u\Big]-\int_{\partial\Omega} z_2u\quad\text{
       over }u\in H^1(\Omega)\text{ with }\int_\Omega u=0\,,
 \end{displaymath}
 and then deduce further regularity for $\omega$ by standard elliptic
 estimates (see \cite{agetal64}).

Once again, it follows by the implicit function theorem that
as long as $\|\nabla\zeta_\mu\|_\infty<1/\sqrt{3}$ the branch $(\zeta_\mu,\mu)$
continuously extends from $(0,0)$. By Lemma \ref{lem:max} the maximum
of $|\nabla\zeta_\mu|$ is attained on the boundary.

To complete the proof we need to demonstrate uniqueness of the above
solutions for every fixed $\mu$. To this end we define the functional
\begin{equation}
  \label{eq:22}
I(\zeta)= - \int_\Omega [\frac{1}{4}(1-|\nabla\zeta|^2)^2 +\nabla_\perp\Phi\cdot\nabla\zeta]\,dx \,, 
\end{equation}
where $\Phi$ is any $C^1$ potential, whose tangential derivative along
$\partial\Omega$ satisfies $\partial\Phi/\partial\tau=\mu j_r$. It can be readily verified that every
critical point of $I$ must satisfy \eqref{eq:15}. Furthermore, $I$ is convex in
\begin{displaymath}
  X = \Big\{ \zeta\in H^1(\Omega)\,\Big| \, \|\nabla\zeta\|_\infty \leq1/\sqrt{3} \;;\;
  \int_\Omega\zeta dx=0\Big\}\,.
\end{displaymath}
and it is strictly convex in the interior of $X$. Thus, for every $\mu$, the functional $I$ can have at most one
critical point, which must therefore lie on the branch of solutions $\zeta_\mu$
whose existence has been established above. 

The proof of \eqref{eq:9.2} has already been established in Lemma \ref{lem:max}.
\end{proof}

\subsection{Solution of the outer problem---$\mathcal{O}\left(\epsilon^2\right)$-term}
Suppose now that there exists a solution for \eqref{eq:15} satisfying 
\begin{equation}
\label{eq:23}
  \|\nabla\zeta\|_\infty<\frac{1}{\sqrt{3}}\,.
\end{equation}
We seek a more accurate estimate of the solution of \eqref{eq:13}
when $\epsilon\ll1$. 
\begin{lemma}
  \label{lem:outer-next-order}
Suppose that a solution of \eqref{eq:15} satisfying \eqref{eq:23}
exists for some $j\in C^{3,\alpha}(\partial\Omega)$ with $0<\alpha<1$. Let $\rho_{out,0}$ and $\chi_{out,0}$ solve
\eqref{eq:14}-\eqref{eq:6.1}. Then there exist
$(\rho_{out,1},\chi_{out,1})\in C^{2,\alpha}(\Omega,\R^2)$, $C>0$, and $\epsilon_0>0$, such that
\begin{displaymath}
  (\rho_o,\chi_o)=(\rho_{out,0},\chi_{out,0})+\epsilon^2(\rho_{out,1},\chi_{out,1}) \,,
\end{displaymath} 
 satisfies
\begin{equation}
  \label{eq:19}
\Big\| \Delta\rho_o +\frac{1}{\epsilon^2}\rho_o(1-\rho_o^2 -\epsilon^2|\nabla\chi_o|^2)\Big\|_\infty+   \|\Div
(\rho_o^2\nabla\chi_o)\|_\infty\leq C\epsilon^2
\end{equation}
for all $0<\epsilon<\epsilon_0$. 
\end{lemma}
\begin{proof}
 Applying a regular perturbation scheme allows us to select
 $(\rho_{out,1},\chi_{out,1})$ as the solution of  
\begin{equation}
\label{eq:2.34}
\left\{
\begin{aligned}
  &\rho_{out,1}(1-3\rho_{out,0}^2-\epsilon^2|\nabla\chi_{out,0}|^2)-
  2\epsilon^2\rho_{out,0}\nabla\chi_{out,0}\cdot\nabla\chi_{out,1} = -\Delta\rho_{out,0} \,, \\
  &\Div\left(\rho_{out,0}^2\nabla\chi_{out,1}+2\rho_{out,0}\rho_{out,1}\nabla\chi_{out,0}\right)=0\,.
  \end{aligned}
  \right.
\end{equation}
Recalling that $\rho_{out,0}^2=1-|\nabla\zeta|^2>2/3$ in $\Omega,$ we obtain with
the aid of \eqref{eq:14} that
\begin{displaymath}
  \rho_{out,1} = \frac{\Delta\rho_{out,0}}{2\rho_{out,0}^2} -
  \frac{\epsilon^2\nabla\chi_{out,0}\cdot\nabla\chi_{out,1}}{\rho_{out,0}} \,.
\end{displaymath}
Substituting this expression into the second equation in
\eqref{eq:2.34} gives the following problem 
\begin{equation}
  \label{eq:24}
  \begin{cases}
    \Div (A(\nabla\zeta)\nabla\zeta_1 ) =  - \Div \Big(\frac{\Delta\rho_{out,0}}{\rho_{out,0}}\nabla\zeta \Big)&
    \text{in } \Omega\,, \\
A(\nabla\zeta)\nabla\zeta_1 \cdot{\bf n} = -\frac{\Delta\rho_{out,0}}{\rho_{out,0}}\nabla\zeta\cdot{\bf n} & \text{on } \partial\Omega \,,
  \end{cases}
\end{equation}
where $A$ is given by \eqref{eq:18} and
\begin{displaymath}
  \zeta_1= \epsilon\chi_{out,1}.
\end{displaymath}
Since $A(\nabla\zeta)\in C^{k-1,\alpha}(\Omega,M(2,2))$ is positive definite by \eqref{eq:23} and the right-hand-side of \eqref{eq:24} is in
$C^{k-4,\alpha}(\Omega)$, we may use Schauder estimates (cf. for instance Theorem 9.3 and
the subsequent Remark 2 in \cite{agetal64}) to conclude the existence of a unique
$C^{k-2,\alpha}$-solution of \eqref{eq:24}. Furthermore, we have that
\begin{equation}
  \label{eq:25}
\|\zeta_1\|_{C^{k-2,\alpha}(\Omega)} \leq C(\Omega,\alpha,k)\,.
\end{equation}

Next, we set 
\begin{equation}
\label{eq:17.5}
  \rho_o = \rho_{out,0}+\epsilon^2\rho_{out,1}, \quad  \chi_o = \chi_{out,0}+\epsilon^2\chi_{out,1} \,.
\end{equation}
With the aid of \eqref{eq:14} and \eqref{eq:2.34} we obtain that
\begin{equation}
\label{eq:26}
  -\Delta\rho_o - \frac{1}{\epsilon^2}\rho_o(1-\rho_o^2 -\epsilon^2|\nabla\chi_o|^2) = g_1,
\end{equation}
where
\begin{multline*}
  g_1 = \epsilon^2\left[-\Delta\rho_{out,1}+
  \rho_{out,0}\left(\rho_{out,1}^2+|\nabla\zeta_1|^2\right) + \right.\\
 \left.\rho_{out,1}\left(2 \rho_{out,0}\rho_{out,1}+2\nabla\zeta\cdot\nabla\zeta_1+\epsilon^2\left(\rho_{out,1}^2+|\nabla\zeta_1|^2\right)\right)\right] \,.
\end{multline*}
If we choose $k\geq3$, then \eqref{eq:25} implies that
\begin{equation}
\label{eq:27}
  \|g_1\|_\infty  \leq C\epsilon^2 \,,
\end{equation}
for some $C>0$. In a similar manner, it also follows that
\begin{equation}
\label{eq:28}
  \Div (\rho_o^2\nabla\chi_o) = g_2,
\end{equation}
where
\begin{equation}
\label{eq:29}
  \|g_2\|_\infty  \leq C\epsilon^2 \,.
\end{equation}
\end{proof}

\section{Inner approximation}
\label{sec:inner-solution}

Next, the outer solution has to be "bridged" to the boundary conditions on
$\partial\Omega$ in \eqref{eq:7} by constructing an appropriate inner solution near the boundary. We
thus consider \eqref{eq:7} in a $\OO(\epsilon)$-thick boundary layer near $\partial\Omega$.  We begin by introducing a curvilinear
coordinate system $(s,t)$ near $\partial\Omega$ by setting
\[(x,y)={\mathbf r}(s)-t\,{\bf n}(s)\] where $t=d(x,\partial\Omega)$ is the
distance function to $\partial\Omega$ that we will assume to be positive in the
interior of $\Omega$. The vector function ${\mathbf r}$ describes $\partial\Omega$
and is parameterized with respect to the arclength $s$ calculated from
some fixed initial point on $\partial\Omega$ in the counterclockwise direction.
The outward unit normal vector ${\bf n}(s)$ to $\partial\Omega$ at ${\mathbf
  r}(s)$ is given by ${\mathbf r}_{ss}(s)=-\kappa(s){\bf n}(s)$, where
$\kappa(s)$ denotes the curvature of $\partial\Omega$ at the point ${\mathbf r}(s)$.
The Jacobian of the transformation $(x,y)\to(s,t)$ is then
\begin{equation}
\label{eq:gf}
  \gf = 1- t\kappa(s) \,.
\end{equation}

Recall that $\sigma=\sigma_0\epsilon^2$. Then, after rescaling
\begin{equation}
\label{eq:30}
 \tau=\frac{t}{\epsilon}, \quad j=-\epsilon J,\quad \varphi=\epsilon^2\phi\,,
\end{equation}
and rewriting the system \eqref{eq:7} in terms of $(s,t)$, we have
\begin{subequations}
\label{eq:31}
\begin{empheq}[left={\empheqlbrace}]{alignat=2}
  &    -\frac{\partial^2\rho}{\partial\tau^2}-\Big(1-\Big|\frac{\partial\upsilon}{\partial\tau}+\frac{\partial\zeta}{\partial t}(s,\epsilon\tau)\Big|^2-\rho^2\Big)\rho =
      &\notag \\ & \qquad
-\epsilon\frac{\kappa}{\gf}\frac{\partial\rho}{\partial\tau}+ \epsilon^2
\Big(\frac{1}{\gf}\frac{\partial}{\partial s}\Big)^2\rho -\epsilon^2\left|\frac{1}{\gf}\left(\frac{\partial\upsilon}{\partial
  s}+\frac{1}{\epsilon}\frac{\partial\zeta}{\partial s}(s,\epsilon\tau)\right)\right|^2 \rho    & \text{ in
} \Omega \\ 
&   -\sigma_0\frac{\partial^2\varphi}{\partial\tau^2} + \rho^2\varphi = -\sigma_0\epsilon\frac{\kappa}{\gf}\frac{\partial\varphi}{\partial\tau}+
   \sigma_0\epsilon^2\Big(\frac{1}{\gf}\frac{\partial}{\partial s}\Big)^2\varphi  & \text{ in } \Omega \\
&    \frac{\partial}{\partial\tau}\left(\rho^2\left(\frac{\partial\upsilon}{\partial\tau}+\frac{\partial\zeta}{\partial t}(s,\epsilon\tau)\right)\right) -\sigma_0\frac{\partial^2\varphi}{\partial\tau^2} =-\sigma_0\epsilon\frac{\kappa}{\gf}\frac{\partial\varphi}{\partial\tau}+\sigma_0\epsilon^2\Big(\frac{1}{\gf}\frac{\partial}{\partial s}\Big)^2\varphi &
\notag \\ 
& \qquad + \epsilon\rho^2\frac{\kappa}{\gf}\left(\frac{\partial\upsilon}{\partial\tau}+\frac{\partial\zeta}{\partial t}(s,\epsilon\tau)\right)
-\epsilon^2\frac{1}{\gf}\frac{\partial}{\partial s}\left(\rho^2\frac{1}{\gf}\left(\frac{\partial\upsilon}{\partial s}+\frac{1}{\epsilon}\frac{\partial\zeta}{\partial s}(s,\epsilon\tau)\right)\right)  & \text{ in } \Omega \\
&    \frac{\partial\rho}{\partial\tau} = 0,\ \frac{\partial\upsilon}{\partial\tau} = \frac{\partial\zeta}{\partial{\bf n}},\  \frac{\partial\varphi}{\partial\tau} = \frac{j}{\sigma_0},\mbox{ and }\int_{\partial\Omega}j=0  & \text{on }  \partial\Omega  \\
 &   \int_\Omega  \upsilon \,dx =0 \,,&
  \end{empheq}
\end{subequations}
where
\begin{equation}
\label{eq:32}
\upsilon(s,\tau)=\chi(s,\tau)-\frac{1}{\epsilon}\zeta(s,\epsilon\tau),
\end{equation}
and (\ref{eq:31}b) is obtained by combining (\ref{eq:7}b) and (\ref{eq:7}c). Note that the second boundary condition in (\ref{eq:31}d) can be written as
\[\frac{\partial\upsilon}{\partial\tau} = \frac{j}{1-|\nabla\zeta|^2}\mbox{ on }\partial\Omega, \]
by taking into account \eqref{eq:15}.

\subsection{Solution of the inner problem---$\mathcal{O}\left(1\right)$-term}
We will attempt to obtain the inner solution
for \eqref{eq:31} through a one-dimensional approximation in terms of
the variable $\tau$ in the direction transverse to $\partial\Omega$.  We thus seek
a solution, denoted by  $(\rho_{i0},\varphi_{i0},\upsilon_{i0})$, for the following problem 
\begin{subequations}
\label{eq:33}
\begin{empheq}[left={\empheqlbrace}]{alignat=2}
&-\rho^{\prime\prime}_{i0} -\left(\rho_r^2(s)-\left|\upsilon_{i0}^\prime+\frac{\partial\zeta}{\partial t}(s,0)\right|^2-\rho_{i0}^2\right)\rho_{i0} = 0 & \ \text{ in } \R_+ \\
&  -\sigma_0\varphi_{i0}^{\prime\prime}+ \rho_{i0}^2\varphi_{i0} = 0 & \text{in } \R_+ \\
&  \left(\rho_{i0}^2\left(\upsilon_{i0}^\prime+\frac{\partial\zeta}{\partial t}(s,0)\right)\right)^\prime - \rho_{i0}^2\varphi_{i0}=0   & \text{in } \R_+ \\
&    \rho_{i0}^\prime(0)=0 & \\
 &    \varphi_{i0}^\prime(0)= \frac{j(s)}{\sigma_0} & \\
 &   \upsilon^\prime_{i0}(0) = -\frac{\partial\zeta}{\partial t}(s,0)\,, & 
  \end{empheq}
\end{subequations}
where
\begin{equation}
\label{eq:rhors}
  \rho_r^2(s) = 1-\Big|\frac{\partial\zeta}{\partial s}(s,0)\Big|^2\,.
\end{equation}
In what follows, we drop the dependence on $s$ for notational
simplicity.   

Adding the second and the third equations in \eqref{eq:33} and
integrating, we find that $\upsilon_{i0}$ can be determined up to a constant
by solving  
\begin{equation}
\label{eq:34}
  \upsilon_{i0}^\prime=\frac{\sigma_0\varphi_{i0}^\prime-j}{\rho_{i0}^2}-\frac{\partial \zeta}{\partial t}(s,0).
\end{equation}
Then $(\rho_{i0},\varphi_{i0})$ satisfy
\begin{equation}
\label{eq:35}
  \begin{cases}
-\rho^{\prime\prime}_{i0} - \Big(\rho_r^2-\frac{(\sigma_0\varphi_{i0}^\prime-j)^2}{\rho_{i0}^4}-\rho_{i0}^2\Big)\rho_{i0} = 0 & \text{in } \R_+ \\
   -\sigma_0\varphi_{i0}^{\prime\prime}+ \rho_{i0}^2\varphi_{i0} = 0 & \text{in } \R_+ \\
    \rho_{i0}^\prime(0)=0 & \\
     \varphi_{i0}^\prime(0)= \frac{j}{\sigma_0}\,. 
  \end{cases}
\end{equation}

As $\tau\to\infty$
we expect that $\left(\varphi_{i0}^\prime,\rho_{i0}^\prime\right)\to(0,0)$ and hence $\lim_{\tau\to\infty}\rho_{i0}=\rho_j,$ where $\rho_j$ solves
\begin{equation}
\label{eq:36}
  \rho_r^2-\frac{j^2}{\rho_j^4}-\rho_j^2  =0 \,.
\end{equation}
It is an easy exercise to show that positive solutions of \eqref{eq:36} exist
as long as
\begin{equation}
  \label{eq:37}
j^2\leq \frac{4}{27}\rho_r^6\,.
\end{equation}
On the other hand, since $j^2 = (1-|\nabla\zeta|^2)^2|\partial\zeta/\partial{\bf n}|^2$ on $\partial\Omega$, it follows that 
\begin{equation}
\label{eq:29.5}
\rho_j^2=\left.(1-|\nabla\zeta|^2)\right|_{\partial\Omega}.
\end{equation} 
\begin{remark}
\label{rmk3}
By \eqref{eq:15} and \eqref{eq:23} we have the following relationship
\begin{displaymath}
  j^2 = (1-|\nabla\zeta|^2)^2\Big|\frac{\partial\zeta}{\partial{\bf n}}\Big|^2 = - (1-|\nabla\zeta|^2)^3
  + \rho_r^2 (1-|\nabla\zeta|^2)^2 < \frac{4}{9}\Big(\rho_r^2-\frac{2}{3}\Big) \leq \frac{4}{27}\rho_r^6,
\end{displaymath}
on $\partial\Omega$. Here the two sides are equal only when $\rho_r=1$, therefore the inequality constraint \eqref{eq:23} is stronger than that in \eqref{eq:37}.
\end{remark}
Without any loss of generality we may assume that $j\leq0$, because we can apply the transformation $(j,\varphi_{i0})\to(-j,-\varphi_{i0})$ otherwise. Following the rescaling \begin{equation}
\label{eq:39}
  \eta = \rho_r\frac{\tau}{\sigma_0^{1/2}}; \quad \vartheta(\eta)=\sigma_0^{1/2}\frac{\varphi_{i0}\left(\sigma_0^{1/2}\rho_r^{-1}\eta\right)}{\rho_r^2}; \quad
  \mu(\eta)=\frac{\rho_{i0}\left(\sigma_0^{1/2}\rho_r^{-1}\eta\right)}{\rho_r}; \quad j_r=\frac{j}{\rho_r^3}\,,
\end{equation} 
the problem \eqref{eq:35} takes the form
\begin{equation}
\label{eq:40}
    \begin{cases}
-\frac{1}{\sigma_0}\mu^{\prime\prime} -\left(1-\frac{\left(\vartheta^\prime-j_r\right)^2}{\mu^4}-\mu^2\right)\mu  = 0 & \text{in } \R_+ ,\\
  -\vartheta^{\prime\prime}+ \mu^2\vartheta = 0 & \text{in } \R_+, \\
    \mu^\prime(0)=0, & \\
     \vartheta^\prime(0)= j_r. & 
  \end{cases}
\end{equation}
We now assume that $\sigma_0$ is large and expand the solution of \eqref{eq:40} in terms of $\sigma_0^{-1}$. We proceed by proving existence of a solution of the corresponding leading order problem and then show that this solution serves as a good approximation for a solution of
\eqref{eq:40}.

\subsubsection{Analysis of the leading order problem in $\sigma_0^{-1}$}
The leading order approximation $(\mu_0,\vartheta_0)$ in $\sigma_0^{-1}$ of the solution $(\mu,\vartheta)$ of \eqref{eq:40} is obtained by neglecting the $\OO(\sigma_0^{-1})$-terms in
\eqref{eq:40}. Thus we look for solutions of
\begin{equation}
\label{eq:41}
  \begin{cases}
    \mu_0^4(1-\mu_0^2) = \left(\vartheta_0^{\prime} -j_r\right)^2 & \text{in } \R_+\,,\\
-\vartheta_0^{\prime\prime} + \mu_0^2\vartheta_0 =0 & \text{in } \R_+\,, \\
\vartheta_0^\prime(0) = j_r \,.
  \end{cases} 
\end{equation}
Because 
\[j_r^2<\frac{4}{27}\]
due to \eqref{eq:37} and Remark \ref{rmk3}, the algebraic equation for $\mu_0^2$ has two distinct positive solutions
whenever $j_r\leq\vartheta_0^\prime\leq0$. We choose the solution for which
\begin{equation}
\label{eq:42}
  \frac{2}{3} < \mu_0^2\leq 1\,.
\end{equation}

We have 
\begin{lemma}
\label{lem:inner-solution-leading}
  There exists a solution $(\mu_0,\vartheta_0)$ of \eqref{eq:41}, where
$(\mu_0-\mu_j,\vartheta_0)\in H^1(\R_+,\R^2)$ and $\mu_j\in\left(\sqrt{2/3},1\right]$ solves
\begin{equation}
\label{eq:43}
   \mu_j^4(1-\mu_j^2)= j_r^2.
\end{equation}
\end{lemma}
\begin{proof}
  We can reduce the system \eqref{eq:41} to a single equation for
$\vartheta_0$ and then obtain $\mu_0$ from $\vartheta_0^\prime-j_r$ by solving an
algebraic 
equation. Indeed, 
let $\mu_j$ satisfy \eqref{eq:43} and set
\begin{displaymath}
  V(t)\overset{def}{=}
  \begin{cases}
     1 & t\leq j_r \\
    \mu_0^2(t) & j_r\leq t\leq 0 \\
    \mu_j^2 & 0\leq t \,,
  \end{cases}
\end{displaymath}
 where $\mu_0(t)$ is determined as the solution of
 \begin{equation*}
    \mu_0^4(t)(1-\mu_0^2(t))=(t-j_r)^2\,,
 \end{equation*}
 satisfying \eqref{eq:42}.
We then look for a solution in $H^1(\R_+)$ of the problem
\begin{equation}
  \label{eq:44}
  \begin{cases}
-w^{\prime\prime} + V(w^\prime)w =0 &  \text{in } \R_+\,, \\
w^\prime(0) = j_r\,.
  \end{cases} 
\end{equation}
Note that the equation in \eqref{eq:44} can be written as
$\frac{w^{\prime\prime}}{V(w^\prime)}=w$ and it has a corresponding variational
formulation. Indeed, set 
\begin{equation*}
  L(p)= \int^p_0 \frac{p-t}{V(t)}\,dt\,,
\end{equation*}
so that $L^{\prime\prime}(p)=\frac{1}{V(p)}\geq 1$ and consider the functional
\begin{equation}
  \label{eq:45}
   \Jg(w)=kw(0)+\int_0^\infty \Big(L(w^\prime)+  \frac{w^2}{2}\Big)\,,\quad w\in H^1(\R_+)\,,
\end{equation}
where $k=-L^\prime(j_r)$.  
Any critical point of $\Jg$ is a solution of \eqref{eq:44} and since
$\Jg$ is strictly convex, there is at most one such critical point---the unique
global minimizer for $\Jg$.  
Thanks to the convexity of $\Jg$, in order to prove existence of a
minimizer it is enough to verify that $\Jg$ is coercive. But this is
evident from the inequality  
\begin{displaymath}
   \Jg(w) \geq \frac{1}{2}\|w\|_{1,2}^2-|k|c_0\|w\|_{1,2}\,,
\end{displaymath}
that can be obtained by observing that $L(p)\geq \frac{p^2}{2}$ 
and appealing to the Sobolev inequality in one
 dimension
 \begin{equation}
\label{eq:46}
   \|v\|_\infty\leq c_0 \|v\|_{1,2}\,,\quad v\in H^1(\R_+)\,.
 \end{equation}
It remains to show that the solution $w$ to \eqref{eq:44} also solves
\eqref{eq:41}. In fact, it would be sufficient to show that
\begin{equation}
  \label{eq:47}
   w^\prime(\eta)\in[j_r,0]\,,\quad\eta\in[0,\infty).
\end{equation}
First we establish that $w\geq0$ in $\R_+$. To this end, because $w\in H^1(\R_+)$
it vanishes at infinity
 \begin{equation}
   \label{eq:48}
\lim_{\eta\to\infty} w(\eta)=0\,.
 \end{equation}
 The function $w$ cannot have a negative minimum in $\R_+$ because $w$
 satisfies \eqref{eq:44} and $V>0$. It then follows from \eqref{eq:43}
 that $w^{\prime\prime}\geq0$ in $\R_+$ and $w^\prime$ is increasing from $j_r$ to
 $0$.
\end{proof}
\begin{remark}
\label{rem:leading-properties}
For future use we note that, taking into account \eqref{eq:42}, we have that
 \begin{equation}
\label{eq:stuff}
-w^{\prime\prime}+\frac{2}{3}w<0\,.
\end{equation}
Then the bound $w(\eta)<w(0)e^{-\sqrt{\frac{2}{3}}\eta}$ for $\eta>0$ can be
deduced via the maximum principle. Further, multiplying the inequality
\eqref{eq:stuff} by $w^\prime$, integrating over $\R_+$ and using
\eqref{eq:48}, we conclude that $w(0)<-\sqrt{3/2}j_r$ and
\begin{equation}
\label{eq:stuff1}
  w(\eta) < \vartheta_m(\eta):=-\sqrt{\frac{3}{2}} j_re^{-\sqrt{\frac{2}{3}}\eta} \,,
\end{equation}
 for all $\eta>0$.
 
Finally, note that we must have $\mu_0(0)=1$ and that by taking
the derivative of the first equation in \eqref{eq:41} it follows that
$\mu^\prime_0(0)=0$.  
\end{remark}

\subsubsection{Existence of a solution for \eqref{eq:40}}
This section is devoted to the proof of the following
\begin{lemma}
\label{lem:inner-solution-exist}
  For a sufficiently large $\sigma_0$ there exists a solution 
$(\mu-\mu_j,\vartheta)\in L^2(\R_+,\R^2)$ of \eqref{eq:40}. Furthermore, there
exists $C>0$ such that for sufficiently large $\sigma_o$ we have 
\begin{equation}
\label{eq:73}
  \sigma_0^{-1/2}\|\mu-\mu_0\|_\infty  +\|\vartheta^\prime-\vartheta_0^\prime\|_\infty + \|\vartheta-\vartheta_0\|_\infty \leq \frac{C}{\sigma_0}\,.
\end{equation}
\end{lemma}
\begin{proof}
  Set
\begin{displaymath}
  \mu_1 =\mu -\mu_0; \quad \vartheta_1=\vartheta - \vartheta_0  \,.
\end{displaymath}
We then use \eqref{eq:41} to rewrite \eqref{eq:40} in the form
\begin{displaymath}
   \begin{cases}
     -\frac{1}{\sigma_0}\mu_1^{\prime\prime}    + (6\mu_0^2-4)\mu_1 +
     2\frac{[1-\mu_0^2]^{1/2}}{\mu_0}\vartheta_1^\prime=
     \frac{1}{\sigma_0}\mu^{\prime\prime}_0 + N_1(\mu_1,\vartheta_1)& \text{in } \R_+\,, \\
   -\vartheta^{\prime\prime}_1+ \mu^2_0\vartheta_1+2\mu_0\mu_1 \vartheta_0  = N_2(\mu_1,\vartheta_1) & \text{in } \R_+\,, \\
    \mu^\prime_1(0)=0\,, & \\
     \vartheta^\prime_1(0)= 0\,. & 
  \end{cases}
\end{displaymath}
Here
\begin{subequations}
\label{eq:50}
  \begin{multline}
  N_1(\mu_1,\vartheta_1) =  -|\vartheta_1^\prime|^2\mu_0^{-3}
  +
  \left(\mu^{-3}-\mu_0^{-3}\right)\left(\mu^4-\left(\vartheta^\prime-j_r\right)^2-\mu^6\right)\\ 
+\mu_1\mu_0^{-3}\left(\mu\mu_0^2+\mu^2\mu_0+\mu^3-3\mu_0^3-\mu^5-\mu^4\mu_0-\mu^3\mu_0^2-\mu^2\mu_0^3-\mu\mu_0^4+5\mu_0^5\right)  
\end{multline}
and
\begin{equation}
  N_2(\mu_1,\vartheta_1) =  -\mu_1^2\vartheta_0 -  (2\mu_0+\mu_1)\mu_1\vartheta_1 \,. 
\end{equation}
\end{subequations}

Suppose that ${\mathbf U}=(U_1,U_2)\in H^2(\R_+,\R^2)$ and define the operators
\begin{subequations}
\label{eq:51}
\begin{empheq}[left={\empheqlbrace}]{alignat=2}
  &   \LL_1({\mathbf U}) = -\frac{1}{\sigma_0}U_1^{\prime\prime}    + (6\mu_0^2-4)U_1 +
     2\frac{[1-\mu_0^2]^{1/2}}{\mu_0}U_2^{\prime}  & \text{ in } &\R_+\,, \\
  & \LL_2({\mathbf U})=-U_2^{\prime\prime}+ \mu^2_0U_2 + 2\mu_0\vartheta_0U_1  & \text{ in }& \R_+ \,.
\end{empheq}
\end{subequations}
Set $\B:D(\B)\to L^2(\R_+,\R^2)$ to be defined by
\begin{equation}
\label{eq:42.5}
  \B\,{\mathbf U}=
\begin{bmatrix}
   \LL_1({\mathbf U}) \\
 \LL_2({\mathbf U}) 
\end{bmatrix}\,,
\end{equation}
where
\begin{displaymath}
  D(\B)=\{{\mathbf U}\in H^2(\R_+,\R^2)\,|\,{\mathbf U}^\prime(0)=0\,\}\,.
\end{displaymath}

 It can be easily verified, that there exists
 $\lambda\leq0$, for which the bilinear form 
\begin{displaymath}
  B({\mathbf U},{\mathbf V}) = \langle{(\B-\lambda\,\mathrm{Id}){\mathbf U},\mathbf V}\rangle=\int_{\R_+}(\B{\mathbf U}-\lambda{\mathbf U})\cdot{\mathbf V}\,,
\end{displaymath}
is coercive in $H^1(\R_+,\R^2)$. As 
\begin{displaymath}
  |B({\mathbf U},{\mathbf V})| \leq \|{\mathbf U}\|_2 \|{\mathbf V}\|_2 \,,
\end{displaymath}
it follows by the Lax-Milgram lemma that $(\B-\lambda\,\mathrm{Id})^{-1}$ is bounded. 

Let then ${\mathbf U}_\lambda=\left(U_{\lambda1},U_{\lambda2}\right)=(\B-\lambda\,\mathrm{Id})^{-1}{\mathbf F}$ where ${\mathbf F}^t=[F_1,F_2]$ and $F_1$ and $F_2$
are both in $L^2(\R_+)$.  By (\ref{eq:51}a) we have that
\begin{displaymath}
  U_{\lambda1}=  -\frac{2[1-\mu_0^2]^{1/2}}{\mu_0(6\mu_0^2-4-\lambda)}U_{\lambda2}^{\prime} +
  \frac{1}{\sigma_0 (6\mu_0^2-4-\lambda)}U_{\lambda1}^{\prime\prime}  +  \frac{1}{(6\mu_0^2-4-\lambda)}F_1 \,.
\end{displaymath}
Substituting this expression into (\ref{eq:51}b) yields
\begin{equation}
\label{eq:52}
  -U_{\lambda2}^{\prime\prime}-
  \frac{4\vartheta_0[1-\mu_0^2]^{1/2}}{6\mu_0^2-4-\lambda}U_{\lambda2}^{\prime} + 
  (\mu^2_0-\lambda)U_{\lambda2} = - \frac{2\mu_0\vartheta_0}{\sigma_0
    (6\mu_0^2-4-\lambda)}(U_{\lambda1}^{\prime\prime}+\sigma_0F_1) +F_2   
\end{equation}

Let
\begin{displaymath}
  G(\eta,\lambda) = \exp \Big\{ \int_0^\eta
  \frac{4\vartheta_0[1-\mu_0^2]^{1/2}}{6\mu_0^2-4-\lambda} \,d\xi \Big\} \,.
\end{displaymath}
By \eqref{eq:42}, the fact that $\lambda\leq0$, and since the function $\vartheta_0\geq0$ in $\R_+$, we have that $G\geq 1$. To establish an upper bound for $G$ we first set
\begin{displaymath}
  \delta = \inf_{\eta\in\R_+}6\mu_0^2(\eta)-4 \,,
\end{displaymath}
and use \eqref{eq:41} and \eqref{eq:42} to note that $\delta>0$ if $j_r^2<4/27$. Since $\vartheta_0\leq\vartheta_m$ on $\R_+$---where $\vartheta_m$ is defined in \eqref{eq:stuff1}---and using \eqref{eq:42} we obtain that
\begin{displaymath}
  G \leq   \exp \Big\{ \frac{4}{\sqrt{3}(\delta-\lambda)}\int_0^\infty 
  \vartheta_m(\eta) \Big\} \,d\eta=\exp\left\{{-\frac{2\sqrt{3}j_r}{\delta-\lambda}} \right\}\,,
\end{displaymath}
i.e., there exists a constant $C>0$, independent of $\delta$ and $\lambda$, such that
\begin{equation}
  \label{eq:53}
G \leq  \exp\left\{{\frac{C}{\delta-\lambda}}\right\}\,.
\end{equation}
To simplify notation, in the remainder of this argument $C$ will denote a generic constant independent of  of $\delta$ and $\lambda$.

We now multiply \eqref{eq:52} by $GU_{\lambda2}$ and integrate by parts to
obtain that
\begin{multline*}
  \|G^{1/2}U_{\lambda2}^\prime\|_2^2 + \|G^{1/2}\sqrt{\mu_0^2-\lambda}U_{\lambda2}\|_2^2 = \Big\langle\frac{2\mu_0\vartheta_0}{\sigma_0
    (6\mu_0^2-4-\lambda)}GU_{\lambda2}^\prime,U_{\lambda1}^\prime\Big\rangle+ \\ \Big\langle\Big(\frac{2\mu_0\vartheta_0}{\sigma_0
    (6\mu_0^2-4-\lambda)}G\Big)^\prime U_{\lambda2},U_{\lambda1}^\prime\Big\rangle -
  \Big\langle\frac{2\mu_0\vartheta_0}{6\mu_0^2-4-\lambda}F_1,GU_{\lambda2}\Big\rangle+ \langle F_2,GU_{\lambda2}\rangle \,.
\end{multline*}
This identity with the aid of \eqref{eq:53} allows us to conclude that
\begin{equation}
\label{eq:54}
   \|U_{\lambda2}^\prime\|_2 + \|U_{\lambda2}\|_2 \leq e^{\frac{C}{\delta-\lambda}}\Big(\|{\mathbf F}\|_2 +
   \frac{1}{\sigma_0}\|U_{\lambda1}^\prime\|_2\Big)\,.
\end{equation}
We next observe that 
\begin{multline*}
  \langle U_{\lambda1},\LL_1({\mathbf U}_\lambda)-\lambda U_{\lambda1}\rangle = \frac{1}{\sigma_0}\|U_{\lambda1}^\prime\|_2^2 +
  \|(6\mu_0^2-4)^{1/2}U_{\lambda_1}\|_2^2 \\  -\lambda\|U_{\lambda1}\|_2^2 
  +\Big\langle2\frac{[1-\mu_0^2]^{1/2}}{\mu_0}U_{\lambda2}^{\prime},U_{\lambda1}\Big\rangle = \langle F_1,U_{\lambda1}\rangle\,.
\end{multline*}
With the aid of \eqref{eq:54} we then obtain,
\begin{displaymath}
  \frac{1}{\sigma_0}\|U_{\lambda1}^\prime\|_2^2 +
  \delta\|U_{\lambda1}\|_2^2 \leq  e^{\frac{C}{\delta-\lambda}}\Big(\|{\mathbf F}\|_2\|U_{\lambda1}\|_2 +
  \frac{1}{\sigma_0}\|U_{\lambda1}^\prime\|_2\|U_{\lambda1}\|_2\Big) \,,
\end{displaymath}
from which we easily conclude that for some $C>0$, 
\begin{displaymath}
  \|U_{\lambda1}\|_2 \leq e^{\frac{C}{\delta-\lambda}}\Big( \frac{1}{\sigma_0}\|U_{\lambda1}^\prime\|_2 + \|{\mathbf F}\|_2\Big)\,.
\end{displaymath}
Consequently, we obtain that
\begin{displaymath}
\frac{1}{\sigma_0^{1/2}}\|U_{\lambda1}^\prime\|_2 + \|U_{\lambda1}\|_2 \leq e^{\frac{C}{\delta-\lambda}}\|{\mathbf F}\|_2\,.
\end{displaymath}
Combining the above with \eqref{eq:54} then yields
\begin{equation}
  \label{eq:55}
\frac{1}{\sigma_0^{1/2}}\|U_{\lambda1}^\prime\|_2 + \|U_{\lambda2}^\prime\|_2 + \|{\mathbf U}\|_2 \leq e^{\frac{C}{\delta-\lambda}}\|{\mathbf F}\|_2\,.
\end{equation}

It is well known that the real resolvent set $\rho(\B)\cap\R$ is
open. Furthermore, by \eqref{eq:55} $ \rho(\B)\cap(-\infty,\delta/2)$ must be
closed. Hence $(-\infty,\delta/2)\subset\rho(\B)$ and, in particular,
$0\in\rho(\B)$. It follows that
\begin{equation}
  \label{eq:56}
  {\mathbf U}=\B^{-1}{\mathbf F}
\end{equation}
is well defined and satisfies
  \begin{equation}
\label{eq:57}
\frac{1}{\sigma_0^{1/2}}\|U_1^\prime\|_2 + \|U_1\|_2 + \|U_2\|_{2,2} \leq
C_\delta\|{\mathbf F}\|_2\,,
\end{equation}
by \eqref{eq:51} and \eqref{eq:55}.

Let now $(\mu_1,\vartheta_1)\in\Wg$ where $\Wg=H^1(\R_+)\times H^2(\R_+)$ is equipped with the
norm
\begin{equation}
\label{eq:58}
   \|(u_1,u_2)\|_{\Wg} = \|u_1\|_2 +
   \frac{1}{\sigma^{1/2}_0}\|u_1^\prime\|_2+ \|u_2\|_{2,2} \,,
\end{equation}
and note that \eqref{eq:57} can be rewritten as
 \begin{equation}
   \label{eq:59}
   \|\B^{-1}{\mathbf F}\|_{\Wg}\leq C_\delta\|{\mathbf F}\|_2\,.
 \end{equation}
Next, let $r$ satisfy
\begin{equation}
  \label{eq:60}
 r=\sigma_0^{-\alpha} \text{ for some }3/4<\alpha<1\,.
\end{equation}
Consider $(\mu_1,\vartheta_1)\in B(0,r)$
   and set 
\begin{equation}
\label{eq:61}
   F_1= \frac{1}{\sigma_0}\mu^{\prime\prime}_0 + N_1(\mu_1,\vartheta_1); \quad F_2= N_2(\mu_1,\vartheta_1)\,,
\end{equation}
with $N_1$ and $N_2$ as defined in \eqref{eq:50}.
Applying the Sobolev inequality \eqref{eq:46} with $v=\mu_1$, using \eqref{eq:60}, yields
\begin{equation}
  \label{eq:62}
  \|\mu_1\|_\infty\leq C\|\mu_1\|_{1,2}\leq C\sigma_0^{1/2}r\leq \sigma_0^{-1/4}\,,
\end{equation}
for sufficiently large $\sigma_0$. In particular,
\begin{equation}
  \label{eq:63}
  \mu\in(\mu_0/2,3\mu_0/2)\,.
\end{equation}
Similarly,
\begin{equation}
  \label{eq:64}
\begin{aligned}
  \|\vartheta_1^\prime\|_\infty&\leq
 C\|\vartheta_1^\prime\|_{1,2}\leq Cr\leq
 \sigma_0^{-3/4}\,,\\
\|\vartheta_1\|_\infty&\leq C \|\vartheta_1\|_{1,2}\leq Cr\leq
 \sigma_0^{-3/4}\,.
\end{aligned}
\end{equation}
Using \eqref{eq:63}-\eqref{eq:64} in \eqref{eq:50} yields,
\begin{equation}
\label{eq:65}
  \begin{aligned}
    |N_1(\mu_1,\vartheta_1)|&\leq c_1|\mu_1|^2+c_2|\vartheta_1^\prime|^2\,,\\
    |N_2(\mu_1,\vartheta_1)|& \leq c_3|\mu_1|^2+c_4|\vartheta_1|^2\,,
  \end{aligned}
\end{equation}
for some constants $c_1,\ldots,c_4$. Now, with the aid of
\eqref{eq:62} and \eqref{eq:64} equation \eqref{eq:65} gives
\begin{equation}
  \label{eq:66}
\begin{aligned}
 \| N_1(\mu_1,\vartheta_1)\|_2&\leq C\left(\|\mu_1\|_2\|\mu_1\|_\infty+
 \|\vartheta_1^\prime\|_2\|\vartheta_1^\prime\|_\infty\right)\\
 &\leq
 C\|\mu_1\|_2(\|\mu_1\|_2+\|\mu_1^\prime\|_2)+C\|\vartheta_1^\prime\|_2(\|\vartheta_1^\prime \|_2+\|\vartheta_1^{\prime\prime}\|_2)\\
&\leq C \|(\mu_1,\vartheta_1)\|_{\Wg}^2(\sigma_0^{1/2}+1)\,.
\end{aligned}
\end{equation}
 By a similar argument also
 \begin{equation}
   \label{eq:67}
    \| N_2(\mu_1,\vartheta_1)\|_2\leq C \|(\mu_1,\vartheta_1)\|_{\Wg}^2(\sigma_0^{1/2}+1)\,.
 \end{equation}
By \eqref{eq:66}--\eqref{eq:67} and \eqref{eq:61} we obtain, for a sufficiently large $\sigma_0$, that
\begin{equation}
\label{eq:68}
  \|{\mathbf F}\|_2 \leq C\Big(\sigma_0^{1/2}r^2+\frac{1}{\sigma_0}\Big) \,.
\end{equation}
We now define the non-linear operator $\A:\Wg\to \Wg$ by
$\A(\mu_1,\vartheta_1)={\mathbf U}$, where ${\mathbf U}$ is defined via
\eqref{eq:56} and ${\mathbf F}$ given by \eqref{eq:61}.

 We first show that $\A:B(0,r)\to B(0,r)$. To this
end we use \eqref{eq:59}, \eqref{eq:60}, and \eqref{eq:68} to obtain that
\begin{equation}
\label{eq:69}
   \|{\mathbf U}\|_{\Wg}
   \leq C\Big(\sigma_0^{1/2}r^2+\frac{1}{\sigma_0}\Big)\leq C\sigma_0^{-1}<r\,,
\end{equation}
for a sufficiently large $\sigma_0$, i.e.,  $\A(\mu_1,\vartheta_1)\in B(0,r)$. 

Finally, we prove that $\A$ is a contraction. Let
$(v_1,w_1)$ and $(v_2,w_2)$ be in $B(0,r)$. Let
${\mathbf V}=(v_1-v_2,w_1-w_2)$. 
A direct computation, using 
\eqref{eq:50}, gives
\begin{multline}
  \label{eq:70}
  | N_1(v_1,w_1) -N_1(v_2,w_2)|+| N_2(v_1,w_1)
  -N_2(v_2,w_2)|\leq\\ C\max(|v_1|,|w_1^\prime|,|w_1|,|v_2|,|w_2^\prime|,|w_2|)(|v_1-v_2|+|w_1-w_2|+|w_1^\prime-w_2^\prime|)\,.
\end{multline}
From \eqref{eq:62}, \eqref{eq:64}, and \eqref{eq:70} we obtain, 
 \begin{equation}
\label{eq:71}
  \|N_1(v_1,w_1)
  -N_1(v_2,w_2)\|_2 + \|N_2(v_1,w_1)
  -N_2(v_2,w_2)\|_2 \leq  C\sigma_0^{-1/4}\|{\mathbf V}\|_{\Wg} \,.
\end{equation}
Finally, applying \eqref{eq:59} and \eqref{eq:71} we get that
\begin{align*}
  \| \A(v_1,w_1)&- \A(v_2,w_2) \|_{\Wg}=\\
&\left\|\B^{-1}\left(\frac{\mu_0^{\prime\prime}}{\sigma_0}+N_1(v_1,w_1),N_2(v_1,w_1)\right)-\B^{-1}\left(\frac{\mu_0^{\prime\prime}}{\sigma_0}+N_1(v_2,w_2),N_2(v_2,w_2)\right)\right\|_{\Wg}\\
&\leq C\sigma_0^{-1/4}\|{\mathbf V}\|_{\Wg}\leq\frac{1}{2}\|{\mathbf V}\|_{\Wg}\,,
\end{align*}
 for large enough $\sigma_0$, i.e., $\A:B(0,r)\to B(0,r)$ is a strict
 contraction.
Applying
Banach Fixed Point Theorem
completes the proof of existence.

 Note that by \eqref{eq:69} we have that
\begin{equation}
\label{eq:72}
  \|\mu-\mu_0\|_2 + \sigma^{-1/2}_0\|\mu^\prime-\mu_0^\prime\|_2 + \|\vartheta-\vartheta_0\|_{2,2} \leq \frac{C}{\sigma_0} \,,
\end{equation}
which easily yields \eqref{eq:73}.
\end{proof}

\subsubsection{Decay estimate for $(\mu,\vartheta)$}
To complete the proof of Proposition \ref{lem:2.2} we need to establish additional properties of $(\mu,\vartheta)$.
\begin{lemma}
\label{lem:inner-solution-properties}
  There exist $C>0$ and $\gamma>0$ such that any solution $(\mu,\vartheta)$ of \eqref{eq:40} satisfies
  \begin{equation}
    \label{eq:141}
|\mu-\mu_j|+|\mu^\prime| + |\vartheta|+|\vartheta^\prime|\leq Ce^{-\gamma\eta}\,.
  \end{equation}
Furthermore,
\begin{equation}
  \label{eq:74}
\mu\geq\mu_j\,.
\end{equation}
\end{lemma}
\begin{proof}
  Since $\vartheta$ can have neither a positive maximum nor a negative
minimum we obtain that $\vartheta$ positive and decreasing, and hence
also that $j_r<\vartheta^\prime<0$. Consequently, if at some point
$\sqrt{2/3}<\mu<\mu_j$, we have that
\begin{displaymath}
  1-\frac{(\vartheta^\prime-j_r)^2}{\mu^4}-\mu^2>1-\frac{j_r^2}{\mu^4}-\mu^2>0\,.
\end{displaymath}
Further, $\mu^{\prime\prime}<0$ so that $\mu$ cannot have a minimum
value between $\sqrt{2/3}$ and $\mu_j$. In view of \eqref{eq:72} and since $\mu_0>\mu_j$ in $\R_+$ by
\eqref{eq:41}, we obtain that for a sufficiently
large $\sigma_0$, the inequality \eqref{eq:74} must be satisfied. 

A standard comparison argument along with the Hopf lemma now shows that
\begin{equation}
\label{eq:bef90}
  \vartheta  \leq -\frac{j_r}{\mu_j}e^{-\mu_j\eta} \,,
\end{equation}
in $\mathbb R_+$. Further, since $\mu<1$, integrating the second equation in \eqref{eq:40} gives
\begin{displaymath}
  \vartheta^\prime(\eta)=-\int_\eta^\infty\mu^2\vartheta\,d\tilde\eta  \geq \int_\eta^\infty\frac{j_r}{\mu_j}e^{-\mu_j\tilde\eta}\,d\tilde\eta=\frac{j_r}{\mu_j^2}e^{-\mu_j\eta} \,.
\end{displaymath}
It follows that
\begin{equation}
\label{eq:135}
  |\vartheta|+\mu_j|\vartheta^\prime| \leq -\frac{2j_r}{\mu_j}e^{-\mu_j\eta} \,,
\end{equation}
To prove exponential decay of $\mu - \mu_j$ we first observe that
\begin{displaymath}
  -\Big(1-\frac{(\vartheta^\prime-j_r)^2}{\mu^4}-\mu^2\Big)\mu >
  (6\mu_j^2-4)(\mu-\mu_j) + Ce^{-\mu_j\eta} \,.
\end{displaymath}
Standard comparison arguments and \eqref{eq:34} complete the proof of
\eqref{eq:141}. 
\end{proof}
The proof of Proposition \ref{lem:2.2} now follows from Lemma
\ref{lem:inner-solution-exist}, Lemma
\ref{lem:inner-solution-properties} and the transformation
\eqref{eq:39}. 

\subsubsection{Estimates on derivatives of $(\mu,\vartheta)$ along the boundary}
For later reference, we need to obtain some estimates on the
derivatives of $\mu$ and $\vartheta$ with respect to $s$, which is merely a
parameter in \eqref{eq:40}. Let then $\tilde\mu=\mu-\mu_j$. Taking the derivative
of \eqref{eq:40} with respect to $s$ yields
\begin{equation}
\label{eq:75}
  \begin{cases}
  \begin{split}
-\frac{1}{\sigma_0}\Big(\frac{\partial{\tilde\mu}}{\partial s}\Big)^{\prime\prime} +
\Big(&6\mu^2-4-\frac{3}{\sigma_0}\frac{\mu^{\prime\prime}}{\mu}\Big) \frac{\partial{\tilde\mu}}{\partial s} +
2\frac{\vartheta^\prime-j_r}{\mu^3}\frac{\partial\vartheta^\prime}{\partial s} \\ & =-
\Big(6\mu^2-4-\frac{3}{\sigma_0}\frac{\mu^{\prime\prime}}{\mu}\Big)
\frac{\partial\mu_j}{\partial s}  +2\frac{\vartheta^\prime-j_r}{\mu^3}\frac{\partial j_r}{\partial s} , \end{split} & \\
   -\Big(\frac{\partial\vartheta}{\partial s}\Big)^{\prime\prime} +
   \mu^2\Big(\frac{\partial\vartheta}{\partial s}\Big)+2\mu\vartheta\frac{\partial\tilde\mu}{\partial s}= -2\mu\vartheta\frac{\partial\mu_j}{\partial s}, &  \\
    \Big(\frac{\partial{\tilde\mu}}{\partial s}\Big)^\prime(0)=0, & \\
    \Big(\frac{\partial\vartheta}{\partial s}\Big)^\prime(0)= \frac{\partial j_r}{\partial s}\,,
  \end{cases}
\end{equation}
after some manipulations where \eqref{eq:40} is used once again.
 
We can now prove the following
\begin{lemma}
\label{lem:diff-s}
  Let $\left(\frac{\partial{\tilde\mu}}{\partial s},\frac{\partial\vartheta}{\partial s}\right)$ denote a solution of \eqref{eq:75}. Then, there exists
  some $C(j_r)>0$ such that,
  \begin{equation}
    \label{eq:76}
  \Big\|\frac{\partial{\tilde\mu}}{\partial s}\Big\|_{C^2(\R_+)}+  \Big\|\frac{\partial\vartheta}{\partial
    s}\Big\|_{C^2(\R_+)}\leq C \,.
  \end{equation}
Furthermore, there exists some $\gamma>0$ such that
\begin{equation}
\label{eq:77}
  \Big|\frac{\partial\vartheta^\prime}{\partial s}\Big| + \Big|\frac{\partial\tilde{\mu}}{\partial s}\Big|\leq C e^{-\gamma\eta} \,.
\end{equation}
\end{lemma}
\begin{proof}
In order to replace the system in \eqref{eq:75} by a system with
homogeneous boundary conditions we change variables and     let
  \begin{displaymath}
    \frac{\partial\tilde{\vartheta}}{\partial s}= \frac{\partial \vartheta}{\partial s}- \frac{\partial j_r}{\partial s}e^{-\mu_j\eta} \,.
  \end{displaymath}
We now represent \eqref{eq:75} in the following manner
\begin{subequations}
\label{eq:78}
\begin{empheq}[left={\empheqlbrace}]{alignat=2}
&\LL_1\left(\frac{\partial{\tilde\mu}}{\partial s}, \frac{\partial \tilde\vartheta}{\partial s}\right)  =f_1  & \text{ in } &\R_+\,,  \\
&  \LL_2\left(\frac{\partial{\tilde\mu}}{\partial s}, \frac{\partial \tilde\vartheta}{\partial s}\right) =f_2  & \text{ in } &\R_+\,, \\ 
&    \Big(\frac{\partial{\tilde\mu}}{\partial s}\Big)^\prime(0)=0\,, & & \\
 &   \Big(\frac{\partial\tilde{\vartheta}}{\partial s}\Big)^\prime(0)=0 \,, & &
\end{empheq}
\end{subequations}
where
\begin{multline*}
  f_1= \left(6[\mu_j^2-\mu^2]+\frac{3}{\sigma_0}\frac{\mu^{\prime\prime}}{\mu}\right)\frac{\partial\mu_j}{\partial s}
  +2\left[\frac{j_r}{\mu_j^3}-\frac{j_r}{\mu^3}+\frac{\vartheta^\prime}{\mu^3}\right]\frac{\partial
    j_r}{\partial s}  + \\
    \left[6(\mu_0^2-\mu^2)+\frac{3}{\sigma_0}\frac{\mu^{\prime\prime}}{\mu}\right]\frac{\partial{\tilde\mu}}{\partial s} +2\left[\frac{[1-\mu_0^2]^{1/2}}{\mu_0}-\frac{\vartheta^\prime-j_r}{\mu^3}
  \right]\frac{\partial\tilde{\vartheta^\prime}}{\partial s} +2\mu_j\frac{\vartheta^\prime-j_r}{\mu^3} \frac{\partial j_r}{\partial
    s}e^{-\mu_j\eta}\,,
\end{multline*}
\begin{displaymath}
  f_2= -2\mu\vartheta\frac{\partial\mu_j}{\partial s} + (\mu_0^2-\mu^2)
  \frac{\partial\tilde{\vartheta}}{\partial s} +   (\mu_j^2-\mu^2)\frac{\partial j_r}{\partial s}e^{-\mu_j\eta}+2(\mu_0\vartheta_0-\mu\vartheta)\frac{\partial{\tilde\mu}}{\partial s} \,,
\end{displaymath}
and the operators $\LL_1$ and $\LL_2$ are as defined in \eqref{eq:51}. Here, in order to obtain the expression for $f_1$, we have also used the derivative with respect to $s$ of the equation \eqref{eq:43}.

Equivalently, we can represent \eqref{eq:78} in the form
\begin{equation}
\label{eq:79}
  \B V = \tilde{\B}V + F\,,
\end{equation}
where $\B$ is given by \eqref{eq:42.5} while
\begin{displaymath}
  V =
  \begin{bmatrix}
    \partial{\tilde\mu}/\partial s  \\ 
\partial \tilde\vartheta/\partial s 
  \end{bmatrix},
  \end{displaymath}
  \begin{displaymath}
  F=
  \begin{bmatrix}
\Big(6[\mu_j^2-\mu^2]+\frac{3}{\sigma_0}\frac{\mu^{\prime\prime}}{\mu}\Big)\frac{\partial\mu_j}{\partial s}
  +2\Big[\frac{j_r}{\mu_j^3}-\frac{j_r}{\mu^3}+\frac{\vartheta^\prime}{\mu^3}\Big]\frac{\partial
    j_r}{\partial s} +2\mu_j\frac{\vartheta^\prime-j_r}{\mu^3} \frac{\partial j_r}{\partial s}e^{-\mu_j\eta}  \\    
 -2\mu\vartheta\frac{\partial\mu_j}{\partial s}  +   (\mu_j^2-\mu^2)\frac{\partial j_r}{\partial s}e^{-\mu_j\eta}
  \end{bmatrix},
\end{displaymath}
and
\begin{displaymath}
   \tilde{\B}=
   \begin{bmatrix}
   6(\mu_0^2-\mu^2)+\frac{3}{\sigma_0}\frac{\mu^{\prime\prime}}{\mu} &
   2\Big[\frac{[1-\mu_0^2]^{1/2}}{\mu_0}-\frac{\vartheta^\prime-j_r}{\mu^3}\Big]\frac{d}{d\eta}
   \\
2(\mu_0\vartheta_0-\mu\vartheta) & (\mu_0^2-\mu^2)
   \end{bmatrix}.
\end{displaymath}
By \eqref{eq:38} and \eqref{eq:40} we have that for some $\gamma_0>0$
\begin{equation}
\label{eq:80}
  \|e^{\gamma_0\eta}F\|_2 <\infty \,.
\end{equation}
In view of \eqref{eq:40},
\eqref{eq:41} and \eqref{eq:73} we have that
\begin{equation}
\label{eq:81}
  \Big|\frac{1}{\sigma_0}\frac{\mu^{\prime\prime}}{\mu}\Big|=
  \Big|1-\frac{(\vartheta^\prime-j_r)^2}{\mu^4}-\mu^2\Big|\leq \frac{C}{\sigma_0^{1/2}} \,.
\end{equation}
Furthermore, by \eqref{eq:41} and \eqref{eq:73} we have that
\begin{equation}
\label{eq:82}
  \Big|\frac{[1-\mu_0^2]^{1/2}}{\mu_0}-\frac{\vartheta^\prime-j_r}{\mu^3}\Big|\leq
  \frac{C}{\sigma_0^{1/2}} \,.
\end{equation}
Hence, with the aid of \eqref{eq:73} we obtain that 
\begin{displaymath}
  \|\tilde{\B}V\|_2 \leq
  \frac{C}{\sigma_0^{1/2}}\Big[\|V\|_2+\Big\|\frac{\partial\tilde{\vartheta^\prime}}{\partial s}\Big\|_2\Big]\leq \frac{C}{\sigma_0^{1/2}}\|V\|_\Wg  \,.
\end{displaymath}
By \eqref{eq:59} we now obtain,
\begin{displaymath}
  \|V\|_\Wg \leq C(\|\tilde{\B}V\|_2+\|F\|_2)\leq \frac{C}{\sigma_0^{1/2}}\|V\|_\Wg +C\|F\|_2\,,
\end{displaymath}
whence
\begin{equation}
\label{eq:83}
   \|V\|_\Wg \leq C\|F\|_2 \,.
\end{equation}
Using the above and a standard ODE regularity argument we can easily
prove \eqref{eq:76}. 

To prove \eqref{eq:77} we take the inner product of (\ref{eq:78}b)
in $L^2(\R_+)$ with $e^{2\gamma\eta}  \frac{\partial\tilde{\vartheta}}{\partial s}$ to obtain,
using the fact that $\mu\geq\mu_j$ by \eqref{eq:74},
\begin{multline*}
  \left\|\left(e^{\gamma\eta}  \frac{\partial\tilde{\vartheta}}{\partial s}\right)^\prime\right\|_2^2 + \left(\mu^2_j-\gamma^2\right) \left\|e^{\gamma\eta}  \frac{\partial\tilde{\vartheta}}{\partial s}\right\|_2^2
  \\ \leq   \left\|e^{\gamma\eta}  \frac{\partial\tilde{\vartheta}}{\partial s}\right\|_2 \left(\left\|e^{\gamma\eta}2\mu\vartheta\left[ \frac{\partial\tilde\mu}{\partial s}+ \frac{\partial\mu_j}{\partial s}\right]\right\|_2+ \left\|e^{\gamma\eta}(\mu_j^2-\mu^2) \frac{\partial j_r}{\partial s}e^{-\mu_j\eta}\right\|_2\right)\,.
\end{multline*}
 Hence, by \eqref{eq:38} and \eqref{eq:83},
\begin{equation}
\label{eq:84}
   \left\|\left(e^{\gamma\eta}  \frac{\partial\tilde{\vartheta}}{\partial s}\right)^\prime\right\|_2^2 + \left\|e^{\gamma\eta}  \frac{\partial\tilde{\vartheta}}{\partial s}\right\|_2^2
   \leq C,
\end{equation}
for a sufficiently small $\gamma$. Taking the inner product of (\ref{eq:78}a)
in $L^2(\R_+)$ with $e^{2\gamma\eta}  \frac{\partial\tilde{\mu}}{\partial s}$ for some
$\gamma<\gamma_0$ yields 
\begin{multline*}
  \frac{1}{\sigma^2_0}\left\|\left(e^{\gamma\eta}  \frac{\partial\tilde{\mu}}{\partial
        s}\right)^\prime\right\|_2^2 + \left(6\mu^2_j-4-\gamma^2\right)
  \left\|e^{\gamma\eta}  \frac{\partial\tilde{\mu}}{\partial s}\right\|_2^2
  \\ \leq   \left\|e^{\gamma\eta}  \frac{\partial\tilde{\mu}}{\partial s}\right\|_2
  \left(C\left\|e^{\gamma\eta} \left(\frac{\partial\tilde{\vartheta}}{\partial
          s}\right)^\prime\right\|_2+ \left\|e^{\gamma\eta}f_1\right\|_2\right)\,. 
\end{multline*}
With the aid of \eqref{eq:80}, \eqref{eq:84}, \eqref{eq:81}, and
\eqref{eq:82}, we than obtain that
\begin{displaymath}
    \frac{1}{\sigma^2_0}\left\|\left(e^{\gamma\eta}  \frac{\partial\tilde{\mu}}{\partial
        s}\right)^\prime\right\|_2^2 + \left\|e^{\gamma\eta}  \frac{\partial\tilde{\mu}}{\partial s}\right\|_2^2
   \leq C\,.
\end{displaymath}
Sobolev embeddings then provide \eqref{eq:77}. Note that
the constant $C$ in \eqref{eq:77} may depend on $\sigma_0$, a fact that
shouldn't be of any concern to us, since in the sequel we keep $\sigma_0$
fixed (though sufficiently large) while letting $\epsilon\to0$. 
\end{proof}
\begin{remark}
\label{rem:2diff-s}
  Repeating the procedure outlined above, one can similarly obtain 
\begin{equation}
\label{eq:85}
   \Big\|\frac{\partial^2\mu}{\partial s^2}\Big\|_{C^2(\R_+)} +  \Big\|\frac{\partial^2\vartheta}{\partial s^2}\Big\|_{C^2(\R_+)} \leq
  C \,, 
\end{equation}
and that
\begin{equation}
\label{eq:86}
  \Big|\frac{\partial^2\tilde{\mu}}{\partial s^2}\Big| + \Big|\frac{\partial^2\vartheta^\prime}{\partial s^2}\Big| \leq Ce^{-\gamma\eta}  \,.
\end{equation}
\end{remark}

Finally, by \eqref{eq:34} we obtain the bound
\begin{displaymath}
  \Big\|\frac{\partial\upsilon^\prime_{i0}}{\partial s}\Big\|_{C^1(\R_+)} +
  \Big\|\frac{\partial^2\upsilon^\prime_{i0}}{\partial s^2}\Big\|_{C^1(\R_+)}\leq C
\end{displaymath}
so that by (\ref{eq:33}c), \eqref{eq:38}, \eqref{eq:76}, and \eqref{eq:77} we have that
\begin{equation}
\label{eq:87}
  \Big|\frac{\partial\upsilon^\prime_{i0}}{\partial s}\Big|\leq Ce^{-\gamma\eta} \,.
\end{equation}

\subsection{Solution of the inner problem---$\mathcal{O}\left(\epsilon\right)$-term}
To prove that the above solution of \eqref{eq:35} is indeed a good
inner approximation, we need to determine the
next order term in the inner expansion. We thus seek
a solution, denoted by  $(\rho_{i1},\varphi_{i1},\upsilon_{i1})$, of the following problem 
\begin{subequations}
\label{eq:88}
\begin{empheq}[left={\empheqlbrace}]{alignat=2}
&-\rho^{\prime\prime}_{i1}
-\left(\rho_r^2-\left|\upsilon_{i0}^\prime+\frac{\partial\zeta}{\partial t}(s,0)\right|^2-3\rho_{i0}^2\right)\rho_{i1}+2 \rho_{i0}\left(\upsilon_{i0}^\prime+\frac{\partial\zeta}{\partial t}(s,0)\right)\upsilon_{i1}^\prime
& \notag \\ 
&  \qquad = R_\rho(s,\tau) \quad &\text{in } \R_+, \\
&  -\sigma_0\varphi_{i1}^{\prime\prime}+ \rho_{i0}^2\varphi_{i1} + 2\rho_{i0}\rho_{i1}\varphi_{i0}=
  \sigma_0\kappa\varphi_{i0}^\prime  \quad &\text{in } \R_+, \\
&
\left(\rho_{i0}^2\upsilon_{i1}^\prime\right)^\prime+2\left(\rho_{i0}\rho_{i1}\left(\upsilon_{i0}^\prime+\frac{\partial\zeta}{\partial t}(s,0)\right)\right)^\prime
-\sigma_0\varphi_{i1}^{\prime\prime}=R_\upsilon(s,\tau)  \quad &\text{in } \R_+, \\ 
   & \rho_{i1}^\prime(0)=0, &  \\
    & \varphi_{i1}^\prime(0)= 0, & \\
    & \upsilon^\prime_{i1}(0) = 0\,,  & 
\end{empheq}
\end{subequations}
where
\begin{multline}
\label{eq:89}
R_\rho(s,\tau):=-\kappa\rho_{i0}^\prime-2\rho_{i0}\tau\left(\upsilon_{i0}^\prime+\frac{\partial\zeta}{\partial t}(s,0)\right)\frac{\partial^2\zeta}{\partial t^2}(s,0) \\ -2\rho_{i0}\left(\kappa\tau\frac{\partial\zeta}{\partial s}(s,0)+\tau\frac{\partial^2\zeta}{\partial s\partial t}(s,0)+\frac{\partial\upsilon_{i0}}{\partial s}\right)\frac{\partial\zeta}{\partial s}(s,0),
\end{multline}
and
\begin{displaymath}
R_\upsilon(s,\tau):=-\kappa{j} - \frac{\partial}{\partial s}\left(\rho^2_{i0}\frac{\partial\zeta}{\partial
    s}(s,0)\right)-\left(\rho_{i0}^2\tau\frac{\partial^2\zeta}{\partial
  t^2}(s,0)\right)^\prime \,.
\end{displaymath}
This system is obtained by first extracting terms of order $\epsilon$ in
(\ref{eq:31}) and using (\ref{eq:34}). 

In what follows, we seek estimates of the right hand sides of the equations (\ref{eq:88}a)-(\ref{eq:88}c). We will need the following set of identities. 
\begin{lemma}
\label{l:1}
Suppose that $\zeta$, $j$ and $\rho_j$ are as defined in \eqref{eq:15} and \eqref{eq:29.5}, respectively. Then 
\begin{equation}
\label{eq:135.1}
\kappa{\left|\frac{\partial\zeta}{\partial s}\right|}^2+\frac{\partial\zeta}{\partial t}\frac{\partial^2\zeta}{\partial t^2}+\frac{\partial\zeta}{\partial s}\frac{\partial^2\zeta}{\partial t\partial s}-\frac{1}{2}\frac{\partial}{\partial t}|\nabla\zeta|^2=0
\end{equation}
and
\begin{equation}
\label{eq:135.2}
\kappa j+\frac{\partial}{\partial s}\left(\rho_j^2\frac{\partial\zeta}{\partial s}\right)+\frac{\partial}{\partial t}\left(\rho_j^2\frac{\partial\zeta}{\partial t}\right)=0
\end{equation}
hold on $\partial\Omega$.
\end{lemma}
\begin{proof}
First, use \eqref{eq:gf} to observe that
\[\frac{\partial}{\partial t}|\nabla\zeta|^2=\frac{\partial}{\partial t}\left(\frac{1}{\gf}\left|\frac{\partial \zeta}{\partial s}\right|^2+\left|\frac{\partial \zeta}{\partial t}\right|^2\right)=\frac{\kappa}{\gf^2}\left|\frac{\partial \zeta}{\partial s}\right|^2+\frac{2}{\gf}\frac{\partial \zeta}{\partial s}\frac{\partial^2\zeta}{\partial t\partial s}+\frac{2}{\gf}\frac{\partial\zeta}{\partial t}\frac{\partial^2\zeta}{\partial t^2},\]
then \eqref{eq:135.1} follows by setting $t=0$. To establish \eqref{eq:135.2} note that
\begin{displaymath}
 \frac{\partial}{\partial
    s}\Big(\rho_j^2\frac{\partial\zeta}{\partial
    s}(s,0)\Big) = -\frac{1}{\gf}\frac{\partial}{\partial t}
  \Big(\gf(1-|\nabla\zeta|^2)\frac{\partial\zeta}{\partial t}\Big)_{t=0},
\end{displaymath}
by \eqref{eq:15}. Taking the derivative on the right hand side of this equation gives \eqref{eq:135.2}.
\end{proof}

We now set
\begin{displaymath}
  \tilde{\rho}_{i1} = \rho_{i1}+\frac{1}{2\rho_j}\frac{\partial}{\partial t}|\nabla\zeta|^2\big|_{(s,0)}\tau
\end{displaymath}
and, taking into account \eqref{eq:34}, obtain that 
\begin{subequations}
\label{eq:90}
\begin{empheq}[left={\empheqlbrace}]{alignat=2}
&-\tilde{\rho}^{\prime\prime}_{i1}
-\left(\rho_r^2-\frac{(\sigma_0\varphi_{i0}^\prime-j)^2}{\rho_{i0}^4}-3\rho_{i0}^2\right)\tilde{\rho}_{i1}+2 \rho_{i0}\left(\upsilon_{i0}^\prime+\frac{\partial\zeta}{\partial t}(s,0)\right)\upsilon_{i1}^\prime
& \notag \\ 
&  \qquad = \tilde{R}_\rho(s,\tau)\quad &\text{in } \R_+, \\
&  -\sigma_0\varphi_{i1}^{\prime\prime}+ \rho_{i0}^2\varphi_{i1} +
2\rho_{i0}\tilde{\rho}_{i1}\varphi_{i0}= \tilde{R}_\varphi
  \quad &\text{in } \R_+, \\
&
\left(\rho_{i0}^2\upsilon_{i1}^\prime\right)^\prime+2\left(\rho_{i0}\tilde{\rho}_{i1}\left[\upsilon_{i0}^\prime+\frac{\partial\zeta}{\partial
      t}(s,0)\right]\right)^\prime
-\sigma_0\varphi_{i1}^{\prime\prime}=\tilde{R}_\upsilon(s,\tau) \quad &\text{in } \R_+, \\ 
   & \tilde{\rho}_{i1}^\prime(0)=-\frac{1}{2\rho_j}\frac{\partial}{\partial t}|\nabla\zeta|^2\big|_{(s,0)}\,, &  \\
    & \varphi_{i1}^\prime(0)= 0, & \\
    & \upsilon^\prime_{i1}(0) = 0\,,  & 
\end{empheq}
\end{subequations}
where
\begin{multline*}
 \tilde{R}_\rho(s,\tau)=-\kappa\rho_{i0}^\prime- 2(\rho_{i0}-\rho_j)\tau
   \frac{\partial\zeta}{\partial t}(s,0)\frac{\partial^2\zeta}{\partial t^2}(s,0) - 2\rho_{i0}\tau\upsilon_{i0}^\prime\frac{\partial^2\zeta}{\partial t^2}(s,0) \\ 
 -2(\rho_{i0}-\rho_j)\left(\kappa\tau\frac{\partial\zeta}{\partial s}(s,0)+\tau\frac{\partial^2\zeta}{\partial s\partial t}(s,0) 
\right)\frac{\partial\zeta}{\partial s}(s,0)-2\rho_{i0}\frac{\partial\upsilon_{i0}}{\partial s}\frac{\partial\zeta}{\partial s}(s,0) \\
-\left(\rho_r^2-\frac{(\sigma_0\varphi_{i0}^\prime-j)^2}{\rho_{i0}^4}-3\rho_{i0}^2+2\rho_j^2\right)\frac{1}{2\rho_j}\frac{\partial}{\partial
  t}|\nabla\zeta|^2\big|_{(s,0)}\tau  \,,
\end{multline*}
\begin{displaymath}
  \tilde{R}_\varphi =  \sigma_0\kappa\varphi_{i0}^\prime + \frac{\rho_{i0}}{\rho_j}\varphi_{i0} \frac{\partial}{\partial t}|\nabla\zeta|^2\big|_{(s,0)}\tau \,,
\end{displaymath}
and
\begin{multline*}
  \tilde{R}_\upsilon(s,\tau)=\frac{\partial}{\partial s}\left(\left(\rho_j^2-\rho^2_{i0}\right) 
    \frac{\partial\zeta}{\partial s}(s,0)\right) \\ +\left(\left[\left(\rho_j^2-\rho_{i0}^2\right) \frac{\partial^2\zeta}{\partial t^2}(s,0)  
  + \frac{1}{\rho_j}\left(\rho_{i0}\upsilon_{i0}^\prime+\left(\rho_{i0}-\rho_j\right)\frac{\partial\zeta}{\partial t}(s,0)\right)\frac{\partial}{\partial
      t}|\nabla\zeta|^2\big|_{(s,0)}\right]\tau\right)^\prime.
\end{multline*}
Here we have used Lemma \ref{l:1} to express $\tilde{R}_\rho$ and $\tilde{R}_\varphi$ in terms of $\rho_{i0}-\rho_j$.

We now have the following
\begin{lemma}
  There exists a solution  $(\tilde{\rho}_{i1},\varphi_{i1},\upsilon_{i1})$ for
  \eqref{eq:90} such that
  \begin{equation}
\label{eq:91}
    e^{\gamma\tau}\left[\left|\tilde{\rho}_{i1}\right|+|\varphi_{i1}| + \left|\upsilon_{i1}^\prime\right|\right] \in L^2(\R_+) \,.  
  \end{equation}
\end{lemma}
\begin{proof}
  We begin by integrating (\ref{eq:90}c) to obtain using \eqref{eq:38} and \eqref{eq:87} that 
  \begin{equation}
  \label{eq:79.5}
     \rho_{i0}^2\upsilon_{i1}^\prime+2\rho_{i0}\tilde{\rho}_{i1}\left(\upsilon_{i0}^\prime+\frac{\partial\zeta}{\partial t}(s,0)\right) -\sigma_0\varphi_{i1}^{\prime} = F(\tau)\,,
  \end{equation}
where
\begin{multline*}
  F(\tau)=- \int_\tau^\infty  \frac{\partial}{\partial
    s}\Big((\rho_j^2-\rho^2_{i0})\frac{\partial\zeta}{\partial
    s}(s,0)\Big) d\tau^\prime \\ +\left[\left(\rho_j^2-\rho_{i0}^2\right) \frac{\partial^2\zeta}{\partial t^2}(s,0)  
  + \frac{1}{\rho_j}\left(\rho_{i0}\upsilon_{i0}^\prime+\left(\rho_{i0}-\rho_j\right)\frac{\partial\zeta}{\partial t}(s,0)\right)\frac{\partial}{\partial
      t}|\nabla\zeta|^2\big|_{(s,0)}\right]\tau\,.
\end{multline*}
Now, using the rescaling (\ref{eq:39}) and setting
\begin{displaymath}
  \varsigma(\eta)=\frac{1}{\rho_r}\tilde{\rho}_{i1}\left(\frac{\sigma_0^{1/2}\eta}{\rho_r}\right)-
  \frac{\sigma_0^{1/2}}{\rho_r}\left(\frac{\partial}{\partial
      t}|\nabla\zeta|^2\Big|_{(s,0)}\right)\frac{1}{\mu_j}e^{-\mu_j\eta}\,,\quad \psi(\eta) =
  \frac{\sigma_0^{1/2}}{\rho_r^2}\varphi_{i1}\left(\frac{\sigma_0^{1/2}\eta}{\rho_r}\right) \,, 
\end{displaymath}
we substitute $\upsilon_{i1}^\prime$ from \eqref{eq:79.5} into (\ref{eq:88}ab) and (\ref{eq:88}de)
yielding with the help of \eqref{eq:34} and \eqref{eq:35} the system 
\begin{equation}
\label{eq:92}
  \begin{cases}
 -\frac{1}{\sigma_0}\varsigma^{\prime\prime}+
 \Big(6\mu^2-4
 -\frac{3}{\sigma_0}\frac{\mu^{\prime\prime}}{\mu}\Big)\varsigma+2\frac{\vartheta^\prime-j_r}{\mu^3}\psi^\prime
= \frac{G(\rho_r^{-1}\sigma_0^{1/2}\eta)}{\rho_r^3} &\text{in } \R_+, \\
  -\psi^{\prime\prime}+ \mu^2\varphi_{i1} + 2\mu\vartheta\varsigma=
  \frac{\sigma_0^{1/2}}{\rho_r}\Big[-2\mu\vartheta\left(\frac{\partial}{\partial
      t}|\nabla\zeta|^2\Big|_{(s,0)}\right)\Big(\eta+\frac{1}{\mu_j}e^{-\mu_j\eta}\Big)+ \kappa\vartheta^\prime\Big]
    &\text{in } \R_+, \\
\varsigma^\prime(0)=0, &  \\
\psi^\prime(0)= 0. & 
  \end{cases}
\end{equation}
Here
\begin{multline}
\label{eq:93}
  G(\tau)=\tilde{R}_\rho - 2\frac{\sigma_0\varphi_{i0}^\prime-j_r}{\rho_{i0}^3}F(\tau) +\frac{\rho_r^2}{\sigma_0^{1/2}}\left(\frac{\partial}{\partial
      t}|\nabla\zeta|^2\Big|_{(s,0)}\right)\mu_je^{-\mu_j\rho_r\sigma_0^{1/2}\tau} \\
-\Big(6\rho_{i0}^2-4\rho_r^2-3\frac{\rho^{\prime\prime}_{i0}}{\rho_{i0}}\Big)\left(\frac{\partial}{\partial
    t}|\nabla\zeta|^2\Big|_{(s,0)}\right)\frac{\sigma_0^{1/2}}{\rho_r\mu_j}e^{-\mu_j\rho_r\sigma_0^{1/2}\tau}\,,
\end{multline}
and the linear operator acting on $(\varsigma,\psi)$ in \eqref{eq:92} is
precisely the same as the one acting on $(\partial\tilde\mu/\partial s,\partial\vartheta/\partial s)$
in \eqref{eq:75}. By \eqref{eq:83}, this operator must have a bounded
inverse for a sufficiently large $\sigma_0$. To complete the proof of
existence, we thus need to show that $G\in L^2(\R_+)$. This, however
can be easily shown since by \eqref{eq:38} and \eqref{eq:77} we have
that
\begin{displaymath}
 |F(\tau)| + |\tilde{R}_\rho| \leq Ce^{-\gamma\sigma_0^{-1/2}\tau} \,.
\end{displaymath}
The lemma is proved.
\end{proof}

\begin{remark}
  Repeating the procedure outlined in Lemma \ref{lem:diff-s} (cf. also
  Remark \ref{rem:2diff-s}), one can similarly obtain 
\begin{equation}
\label{eq:94}
  \Big\|\frac{\partial\rho_{i1}}{\partial s}\Big\|_{C^2(\R_+)} +
  \Big\|\frac{\partial\varphi_{i1}}{\partial s}\Big\|_{C^2(\R_+)}+ 
 \Big\|\frac{\partial^2\rho_{i1}}{\partial s^2}\Big\|_{C^2(\R_+)} +
 \Big\|\frac{\partial^2\varphi_{i1}}{\partial s^2}\Big\|_{C^2(\R_+)} \leq 
  C \,, 
\end{equation}
and that
\begin{equation}
\label{eq:95}
 \Big|\frac{\partial^2\tilde{\rho}_{i1}}{\partial s^2}\Big| +
 \Big|\frac{\partial^2\varphi_{i1}^\prime}{\partial s^2}\Big| + 
 \Big|\frac{\partial^2\tilde{\rho}_{i1}}{\partial s^2}\Big| +
 \Big|\frac{\partial^2\varphi_{i1}^\prime}{\partial s^2}\Big| \leq Ce^{-\gamma\eta}  \,. 
\end{equation}
\end{remark}
We can now set
\begin{equation}
\label{eq:107}
  \rho_i = \rho_{i0}+\epsilon\rho_{i1}, \quad   \varphi_i=\varphi_{i0}+\epsilon\varphi_{i1}, \quad  \upsilon_i = \upsilon_{i0}+\epsilon\upsilon_{i1}.
\end{equation}

We conclude this section by the following auxiliary estimate which
will be used in the next section.
\begin{lemma}
\label{lem:inner-solution-1-quadratic}
  Suppose that the outer solution satisfies 
  \begin{equation}
  \label{eq:beta}
  \|\nabla\zeta\|_\infty^2<\beta:=5/14-\sqrt{65}/70.
   \end{equation}
   Let
  \begin{subequations}
\label{eq:133}
    \begin{align}
&a(\eta) = -\rho_r^6\mu^2\vartheta^2& \\      
&b(\eta) =
\rho_r^2\mu^2[3\rho_r^2\mu^2-1-3(|\zeta_s(0,s)|^2+|\chi_{i0}^\prime|^2)]+\frac{1}{4}\rho_r^4\vartheta^2& \\
&c(\eta) =- \frac{1}{4}[3\rho_r^2\mu^2-1+|\zeta_s(0,s)|^2+|\chi_{i0}^\prime|^2] \,,&
    \end{align}
  \end{subequations}
where
\begin{equation}
\label{eq:chi-i0}
  \chi_{i0}^\prime=\rho_r\frac{\vartheta^\prime-j_r}{\mu^2}\,.
\end{equation}
Then, for a sufficiently large $\sigma_0$ the quadratic polynomial
$Q(z,\eta)=az^2+bz+c$ has two distinct real roots for all $\eta\in\R_+$. Denote these roots by 
$z_1(\eta,\nabla\zeta(0,s))$ and $z_2(\eta,\nabla\zeta(0,s))$ and suppose that $z_1<z_2$. If we set
\begin{equation}
  \label{eq:129}
z_m=\hspace{-1.5em}\sup_{
  \begin{subarray}{c}
    \eta\in\R_+ \\
    |\nabla\zeta(0,s)|^2 < \beta 
  \end{subarray}}
\hspace{-1.5em}z_1(\eta), \qquad 
z_M=\hspace{-1.5em}\inf_{
  \begin{subarray}{c}
    \eta\in\R_+ \\
    |\nabla\zeta(0,s)|^2 < \beta
  \end{subarray}}
\hspace{-1.5em}z_2(\eta),
\end{equation}
then $z_M>z_m$.
\end{lemma}
\begin{proof}
For convenience we set
\begin{equation}
\label{eq:139}
  P^2 =|\zeta_s(0,s)|^2+|\chi_{i0}^\prime|^2\,.
\end{equation}
By \eqref{eq:rhors},\eqref{eq:chi-i0} and \eqref{eq:139} we have 
\begin{equation}
\label{eq:est-dif}
\rho_r^2\mu^2-(1-P^2)=\rho_r^2\big(\mu^2-1+\frac{(\vartheta^\prime-j_r)^2}{\mu^4}\big).
\end{equation}
Using \eqref{eq:41} and (\ref{eq:73}) in \eqref{eq:est-dif} we may conclude the existence of
$C>0$ such that for a sufficiently large $\sigma_0$ we have 
\begin{equation}
\label{eq:dif1P2}
  \|\rho_r^2\mu^2-(1-P^2)\|_\infty \leq \frac{C}{\sigma_0}\,.
\end{equation}
In the sequel we will use the notation $A\cong B$ to indicate that $|A-B|\leq \frac{C}{\sigma_0}$. From \eqref{eq:dif1P2} we get that 
 $\rho_r^2\mu^2\cong 1-P^2$ and use \eqref{eq:133} to get that 
  \begin{subequations}
  	\label{eq:subeq}
 \begin{align}
 a&\cong -\rho_r^4\vartheta^2(1-P^2),\\
 b&\cong (1-P^2)(3(1-P^2)-1-3P^2)+\frac{1}{4}\rho_r^4\vartheta^2=(1-P^2)(2-6P^2)+\frac{1}{4}\rho_r^4\vartheta^2,\\
 c&\cong  -\frac{1}{4}\big(3(1-P^2)-(1-P^2)\big)=-\frac{1}{2}(1-P^2).
 \end{align}
\end{subequations}
  Using \eqref{eq:subeq} it can be easily verified that
  \begin{equation}
\label{eq:134}
    \left|b^2 - 4ac -([2(1-3P^2)(1-P^2)]^2 - \rho_r^4\vartheta^2(1-P^2)(1+P^2)+ \rho_r^8\vartheta^4/16)\right| \leq \frac{C}{\sigma_0}\,.
  \end{equation}
Consequently, whenever
\begin{equation}
\label{eq:138}
  \inf_{\eta\in\R_+}\big([2(1-3P^2)]^2(1-P^2)- \rho_r^4\vartheta^2(1+P^2)\big) >0
\end{equation}
then $Q(z)$ would have two distinct roots for a sufficiently large $\sigma_0$. 

By the monotonicity of $\vartheta$ and the fact that $\mu>\mu_j,$ we obtain that 
\begin{equation}
\label{eq:P2lt}
  P^2 \leq 1-\rho_r^2 + \rho_r^2\frac{j_r^2}{\mu_j^4} = 1-\rho_j^2 =
  |\nabla\zeta(0,s)|^2 \,.
\end{equation}
 Above we used the relation 
 \begin{equation}
 \label{eq:rhojr}
 \rho_j=\rho_r\mu_j
 \end{equation}
 that it is expected to hold in view of \eqref{eq:39} and the fact
 that $\rho_j$ is the asymptotic limit of $\rho_{i0}$ while $\mu_j$ is the
 asymptotic limit of $\mu$ (which follows from Lemma
 \ref{lem:inner-solution-exist} and Proposition \ref{lem:2.2}). One can
 also verify \eqref{eq:rhojr} directly from \eqref{eq:36},
 \eqref{eq:39}, and (\ref{eq:43}).
 
 Next, by \eqref{eq:bef90},\eqref{eq:43}  and \eqref{eq:rhojr}    it follows that
\begin{displaymath}
  \rho_r^4\vartheta^2 \leq \rho_r^4\frac{j_r^2}{\mu_j^2}=\rho_j^2(\rho_r^2-\rho_j^2)\leq\rho_j^2(1-\rho_j^2)\,.
\end{displaymath}
From the above we deduce that \eqref{eq:138} is satisfied if 
\begin{displaymath}
  (1-\rho_j^2) (2-\rho_j^2)< 4(3\rho_j^2-2)^2\,,
\end{displaymath}
which clearly holds whenever
\begin{equation}
\label{eq:rho-beta} 
 1-\rho_j^2 <\beta,
\end{equation}
see \eqref{eq:beta}.
Let then $z_1(\eta)$ and $z_2(\eta)$ denote the roots of $Q$, where
$z_1<z_2$. Then,
\begin{displaymath}
  z_1=-\frac{b}{2a}\Big[1-\sqrt{1-\frac{4ac}{b^2}}\Big]\leq -\frac{2c}{b}\,.
\end{displaymath}
Hence, since $P^2<\beta$ (by \eqref{eq:P2lt} and \eqref{eq:rho-beta}), we obtain
\begin{equation}
\label{eq:137}
  z_m \leq\sup_{\eta\in\R_+}\frac{1}{2(1-3P^2)}+O(\frac{1}{\sigma_0})<\frac{1}{2(1-3\beta)}+O(\frac{1}{\sigma_0}) \,.
\end{equation}
In a similar manner we obtain
\begin{displaymath}
  z_2 =-\frac{b}{2a}\Big[1+\sqrt{1-\frac{4ac}{b^2}}\Big]=-\frac{b}{a}-\left(\frac{-b}{2a}\right)\Big[1-\sqrt{1-\frac{4ac}{b^2}}\Big]\geq -\frac{b}{a}+\frac{2c}{b}\,.
\end{displaymath}
Hence, for some $C>0$ and sufficiently large $\sigma_0$,
\begin{equation}
\label{eq:136}
  z_M \geq \frac{1}{4\rho_j^2} + \frac{2-6P^2}{\rho_r^4\vartheta^2} - z_m -\frac{C}{\sigma_0}\,.
\end{equation}

Standard comparison arguments show that
\begin{displaymath}
  \vartheta\geq-j_re^{-\eta}\,.
\end{displaymath}
Consequently,
\begin{displaymath}
  \vartheta^\prime= -\int_\eta^\infty\mu^2\vartheta\,d\tilde\eta \leq \mu_j^2j_re^{-\eta}
\end{displaymath}
We now use the above together with \eqref{eq:bef90} to obtain that
\begin{equation}
\label{eq:inequality}
  \frac{2-6P^2}{\vartheta^2} \geq
  \frac{2-6[|\zeta_s(0,s)|^2+j^2_r\mu_j^{-4}\rho_r^2(1-\mu_j^2e^{-\eta})^2]}{j_r^2\mu_j^{-2}e^{-2\mu_j\eta}}\,.
\end{equation}
Next we claim that  the right-hand-side of \eqref{eq:inequality}
 is monotone increasing with respect to $\eta$. Indeed, setting 
\begin{equation}
	\label{eq:AB}
A= \frac{2-6|\zeta_s(0,s)|^2}{j_r^2\mu_j^{-2}}~\text{ and }~B=6\mu_j^{-2}\rho_r^2,
\end{equation}
 we rewrite it as
 \begin{equation}
 \label{eq:feta}
 f(\eta):=
 Ae^{2\mu_j\eta}-Be^{2\mu_j\eta}(1-\mu_j^2e^{-\eta})^2=
 (A-B)e^{2\mu_j\eta}+\mu_j^2B\Big[2e^{(2\mu_j-1)h}-\mu_j^2e^{-2(1-\mu_j)h}\Big].  
\end{equation}
Now we note that $A>B$. To see this, we use
\eqref{eq:43},\eqref{eq:rhors} and \eqref{eq:rhojr} to write
  \begin{multline}
  \label{eq:A-B}
  j_r^2\mu_j^{-2}(A-B)=2-6|\zeta_s(0,s)|^2-6j_r^2\mu_j^{-4}\rho_r^2\\=2-6(1-\rho_r^2)-6(1-\mu_j^2)\rho_r^2=6\mu_j^2\rho_r^2-4=6\rho_j^2-4\geq 2-6\beta>0.
  \end{multline}
  Here we used that $\rho_j^2=1-|\nabla\zeta(0,s)|^2\geq 1-\beta$ by \eqref{eq:29.5} and
  \eqref{eq:rho-beta}, and finally \eqref{eq:beta}.  In particular it
  follows that $A>B$.  The monotonicity of $f(\eta)$ follows by using the
  inequalities $A>B$ and $\mu_j<1<2\mu_j$ in \eqref{eq:feta}.
  \par By the above and \eqref{eq:inequality} we get that 
\begin{multline}
\label{eq:gtf0}
  \frac{2-6P^2}{\vartheta^2} \geq f(0)=\big(2-6[|\zeta_s(0,s)|^2+j^2_r\mu_j^{-4}\rho_r^2(1-\mu_j^2)^2]\big)\frac{\mu_j^2}{j^2_r}\\
    =A-B+6\rho_r^2(2-\mu_j^2).
\end{multline}
By \eqref{eq:gtf0} and \eqref{eq:A-B} we obtain that 
\begin{equation}
\label{eq:2-6P}
 \frac{2-6P^2}{\rho_r^4\vartheta^2} \geq\frac{6}{\rho_r^2}(2-\mu_j^2)\geq 6,
\end{equation}
 where in the last inequality we used the inequalities  $\rho_r^2\leq1$ and $\mu_j\leq1$.
 
 Finally, substituting \eqref{eq:2-6P} into \eqref{eq:136} yields,
 with the help of \eqref{eq:137},
\begin{displaymath}
 z_M-z_m\geq 6-\frac{1}{1-3\beta}-\frac{C}{\sigma_0}=\frac{5-18\beta}{1-3\beta}-\frac{C}{\sigma_0}>0
\end{displaymath}
 for sufficiently  large $\sigma_0$, since $\beta<5/18<1/3$ by \eqref{eq:beta}. 
\end{proof}

\section{Uniform approximation}
\label{sec:unif-appr}

We can now construct an approximate solution for \eqref{eq:6}. To this end,
let $\Upsilon\in C^\infty(\R,[0,1])$ denote a cutoff function satisfying 
\begin{equation}
\label{eq:96}
  \Upsilon(x) =
  \begin{cases}
    1, & x<\frac{1}{2}, \\
    0, & x>1,
  \end{cases}
\qquad \text{ where } |\Upsilon^\prime|\leq 4 \,.
\end{equation}
Let, for some $x\in\Omega$, $t=d(x,\partial\Omega)$, $\tau=t/\epsilon$, and $s$ denote the
arclength calculated from some fixed initial point on $\partial\Omega$ to the
projection of $x$ on $\partial\Omega$ in the counterclockwise direction.
Further, let $\delta=\epsilon^{\iota}$ for some $0<\iota<1$ and set, for all $(x,y)\in\Omega$,
\begin{subequations} 
\label{eq:97}
\begin{equation}
\phi_0(x,y) =  \frac{1}{\epsilon^2}\varphi_i\big(s,\tau\big)\Upsilon(t/\delta)
\end{equation}
\vspace*{1.5ex}
\begin{equation}
\rho_0(x,y) = \rho_o + \big[\rho_i\big(s,\tau\big) - 
\rho_a(s,t)\big]\Upsilon(t/\delta)  \,,
\end{equation}
and
\begin{equation}
\chi_0(x,y) = \chi_o  + \upsilon_i\big(s,\tau\big)\Upsilon(t/\delta).
\end{equation}
Here
\begin{equation}
\label{eq:102}
  \rho_a(s,t) =\left[1-|\nabla\zeta(s,0)|^2\right]^{1/2} +
  \left.\left(\frac{\partial}{\partial t}\left[1-|\nabla\zeta(s,t)|^2\right]^{1/2}\right)\right|_{t=0} t \,,
\end{equation}
while $\rho_o$ and $\chi_o$ are as defined in \eqref{eq:17.5}.
\end{subequations}
We shall attempt to prove Theorem \ref{thm:stationary2} using the
Banach fixed point theorem. To this end we set
$$(\rho_1,\phi_1,\chi_1)=(\rho-\rho_0,\epsilon^2[\phi-\phi_0],\epsilon[\chi-\chi_0]),$$ and
\begin{equation}
  \label{eq:106}
(\phi_0,\chi_0)=(\epsilon^{-2}\tilde{\phi_0},\epsilon^{-1}\tilde{\chi}_0).
\end{equation}
We then
rewrite \eqref{eq:7} in the following form
  \begin{subequations}
\label{eq:98}
\begin{empheq}[left={\empheqlbrace}]{alignat=2}
  -\Big(\Delta +\frac{1}{\epsilon^2}(1-&3\rho_0^2-|\nabla\tilde{\chi}_0|^2)\Big)\rho_1 +
  \frac{2\rho_0}{\epsilon^2}\nabla\tilde\chi_0\cdot\nabla\chi_1
  \notag   \\ & = -h_1 -  \frac{1}{\epsilon^2}\Big[|\nabla\chi_1|^2\rho_0 +  2\rho_1\nabla\tilde{\chi}_0\cdot\nabla\chi_1  \nonumber \\ & + |\nabla\chi_1|^2\rho_1+(3\rho_0+\rho_1)\rho_1^2\Big]  \\
  \Div(\rho_0^2\nabla\chi_1) + & 2\Div \big(\rho_1\rho_0\nabla\tilde{\chi}_0\big) - \sigma_0\epsilon\Delta\phi_1
  \notag \\ & = -\epsilon \Div H_2 - \Div\big(\rho_1^2\nabla\tilde{\chi}_0\big)  - \Div \big(\rho_1(2\rho_0+\rho_1)\nabla\chi_1\big) \\[1.5ex]
  \sigma_0\epsilon^2\Delta\phi_1 -\rho_0^2&\phi_1 - 2\rho_0\tilde{\phi}_0\rho_1 = -\epsilon^2h_3 +
  \rho_1(2\rho_0+\rho_1)\phi_1 + \rho_1^2\tilde{\phi}_0 \,.
\end{empheq}
where
\begin{align}
  & h_1 = -\Delta\rho_0 - \frac{1}{\epsilon^2}(1-|\nabla\tilde{\chi}_0|^2-\rho_0^2)\rho_0 \\
  & H_2 = \frac{1}{\epsilon}\rho_0^2\nabla\tilde\chi_0 - \sigma_0\nabla\tilde\phi_0 \\
  & h_3 = \sigma_0\Delta\tilde\phi_0- \frac{1}{\epsilon^2}\rho_0^2\tilde{\phi}_0 \,.
\end{align}
  \end{subequations}

The first step towards establishing Theorem \ref{thm:stationary2} is
to show that $h_1$, $H_2$, and $h_3$ are small in the limit $\epsilon\to0$.
\begin{lemma}
   Let
 \begin{displaymath}
   \Omega_r = \{ x\in\Omega\,|\, d(x,\partial\Omega)\geq r\,\}\,.
 \end{displaymath}
Then,
\begin{subequations}
\label{eq:49}
  \begin{align}
    &   \|h_1\|_2 \leq C\epsilon^{\iota/2-2(1-\iota)}  \,,& \\
 & \|\Div H_2\|_{L^2(\Omega_\delta)} \leq C\epsilon^2  \,, &\\
  &  \|H_2\|_{L^\infty (\overline{\Omega\setminus\Omega_\delta)}} \leq C_\iota\epsilon^{1+\iota/2-2(1-\iota)}&
  \\
   &   \|\Div H_2\|_2 \leq C_\iota\epsilon^{\iota/2-2(1-\iota)} \,,  & \\
& \|h_3\|_2 \leq C_\iota\epsilon^{\iota/2-2(1-\iota)}\,,& \\
& \|\Delta\tilde{\chi}_0\|_2 \leq \frac{C}{\epsilon^{1/2}} \,.
\end{align}
\end{subequations}
\end{lemma}
\begin{proof}
  {\em (i).} We begin by seeking an estimate for $\|h_1\|_2$.  
By \eqref{eq:38}, \eqref{eq:91}, and \eqref{eq:95} we
have that for all $x\in\Omega_{\delta}$
\begin{displaymath}
  h_1=g_1 + \OO\big(e^{-\epsilon^{-(1-\iota)}}\big) \,,
\end{displaymath}
where $g_1$ is given by \eqref{eq:26}.
Hence, by \eqref{eq:27} we have
\begin{equation}
\label{eq:99}
  \|h_1\|_{L^2(\Omega_\delta)} \leq C\epsilon^2  \,.
\end{equation}

Next, we estimate $h_1$ in a $\delta$-neighborhood of $\partial\Omega$. In this
neighborhood we may write
\begin{equation}
\label{eq:100}
  h_1= -\Delta\rho_i-\frac{\rho_i}{\epsilon^2}\left[1-\rho_i^2-\epsilon^2\left|\nabla\upsilon_i+\frac{1}{\epsilon}\nabla\zeta_a\right|^2\right]+\tilde{h}_1 \,,
\end{equation}
where
\begin{multline*}
  \tilde{h}_1 =
  -\Delta(\rho_o-\rho_a)-\frac{\rho_o-\rho_a}{\epsilon^2}\left[1-3\rho_i^2-3\rho_i(\rho_o-\rho_a)-(\rho_o-\rho_a)^2-\epsilon^2|\nabla\left(\chi_o+\upsilon_i\right)|^2\right]\\{+}2\rho_i\nabla\left(\chi_o-\frac{\zeta_a}{\epsilon}\right)\cdot\nabla\left(\upsilon_i+\frac{1}{\epsilon}\zeta_a\right){+}
 \rho_i\left|\nabla\left(\chi_o-\frac{\zeta_a}{\epsilon}\right)\right|^2,
\end{multline*}
and
\begin{equation}
 \zeta_a = \zeta(s,0)+t\frac{\partial\zeta}{\partial t}(s,0)+\frac{t^2}{2}\frac{\partial^2\zeta}{\partial
   t^2}(s,0)\,.
\end{equation}
Since $\rho_a$ and $\zeta_a$ are the respective Taylor expansion of $\rho_{out,o}$
and $\chi_{out,0}$ near the boundary, it easily follows that
\begin{equation}
  \label{eq:101}
\| \tilde{h}_1 \|_{L^\infty (\Omega\setminus\Omega_\delta)} \leq C\epsilon^{-2(1-\iota)} \,.
\end{equation}

To complete the estimate of $h_1$ it remains necessary to bound
$h_1-\tilde{h}_1$ in \eqref{eq:100}.  To this end we note first that
\begin{displaymath}
  \Delta\rho_i =
  \Big(\frac{1}{\gf}\frac{\partial}{\partial s}\Big)^2\rho_i 
+\frac{1}{\epsilon^2}\rho_i^{\prime\prime} - \frac{\kappa}{\epsilon\gf}\rho_i^\prime  =
\frac{1}{\epsilon^2}\rho_{i0}^{\prime\prime} + \frac{1}{\epsilon}\Big[\frac{\kappa}{\gf}\rho_{i0}^\prime+\rho_{i1}^{\prime\prime}\Big]+\OO(1)\,.
\end{displaymath}
Furthermore, by (\ref{eq:97}c)
\begin{multline*}
\epsilon^2|\nabla\chi_i|^2 = |\chi_i^\prime|^2 + \frac{\epsilon^2}{\gf^2}\Big|\frac{\partial\chi_i}{\partial s}\Big|^2=
|\chi_{i0}^\prime|^2 + \Big|\frac{\partial\zeta(s,0)}{\partial s}\Big|^2 \\
+2\epsilon\bigg[\chi_{i0}^\prime\chi_{i1}^\prime+\bigg(\frac{\partial^2\zeta}{\partial t\partial s}(s,0)\tau+\int_\tau^\infty\Big[\frac{\partial\chi_{i0}^\prime(\tau^\prime,s)}{\partial s} -
\frac{\partial^2\zeta}{\partial t\partial s}(s,0)\Big]\,d\tau^\prime\bigg)\frac{\partial\zeta(s,0)}{\partial s} \\ +\kappa\tau \Big|\frac{\partial\zeta(s,0)}{\partial s}\Big|^2\bigg] 
+ \OO(\delta^2) \,.
\end{multline*}
Hence, by \eqref{eq:38} \eqref{eq:35}, \eqref{eq:34}, \eqref{eq:88},
and \eqref{eq:91} we obtain that
\begin{displaymath}
   \Big|-\Delta\rho_i-\frac{\rho_i}{\epsilon^2}[1-\rho_i^2-\epsilon^2|\nabla\chi_i|^2]\Big| \leq C\epsilon^{-2(1-\iota)}
  \quad \forall x\in\Omega\setminus\Omega_\delta \,.
\end{displaymath}
Combining the above with \eqref{eq:100} and \eqref{eq:101} yields
(\ref{eq:49}a). 

\medskip\noindent
{\em (ii).} Next, we obtain estimates on $\|H_2\|_2$, where $H_2$ is given by
(\ref{eq:98}e). We have that for all $x\in\Omega_{\delta}$ 
\begin{displaymath}
  \Div H_2=g_2 + \OO\big(e^{-\epsilon^{-(1-\iota)}}\big) \,,
\end{displaymath}
where $g_2$ is given by \eqref{eq:28}.  Hence, we may conclude (\ref{eq:49}b).
\begin{equation}
\label{eq:103}
  \|\Div H_2\|_{L^2(\Omega_\delta)} \leq C\epsilon^2  \,.
\end{equation}

Setting
\begin{equation}
\label{eq:104}
  H_2 = \rho_i^2\nabla\chi_i - \sigma_0\nabla\varphi_i +\tilde{H}_2 
\end{equation}
in $\Omega\setminus\Omega_\delta$, where
\begin{displaymath}
  \tilde{H}_2 = \rho_0^2\nabla\Big(\chi_o - \frac{\zeta_a}{\epsilon}\Big) + \big[\rho_0^2-\rho_i^2\big]\nabla\chi_0 \,,
\end{displaymath}
we deduce the estimate 
\begin{equation}
\label{eq:105}
\| \tilde{H}_2 \|_{L^\infty (\Omega\setminus\Omega_\delta)} \leq C\epsilon^{1-2(1-\iota)} \,.
\end{equation}
It follows that
\begin{displaymath}
  |\rho_i^2\nabla\chi_i - \sigma_0\Delta\varphi_i|  \leq C\epsilon^{1-2(1-\iota)}
  \quad \forall x\in\Omega\setminus\Omega_\delta \,,
\end{displaymath}
which together with \eqref{eq:105} yields (\ref{eq:49}c). Furthermore,
with the aid of (\ref{eq:49}b) we obtain (\ref{eq:49}d).

\medskip\noindent
{\em (iii).} We now derive an estimate for $\|h_3\|_2$, where $h_3$ is given by
(\ref{eq:98}f).  By \eqref{eq:97} we readily obtain that for $x\in\Omega_\delta$
we have
\begin{displaymath}
  h_3=\OO(e^{-\epsilon^{-(1-\iota)}}) \,.
\end{displaymath}
In $\Omega\setminus\Omega_\delta$ it can be easily verified that 
\begin{displaymath}
  h_3 = \sigma_0\Delta\varphi_i - \rho_i^2\varphi_i + (\rho_o-\rho_a)(\rho_o-\rho_a+2\rho_i)\varphi_i \,.
\end{displaymath}
Once again, the same argument used to derive (\ref{eq:49}a), yields
(\ref{eq:49}e).

\noindent{\em (iv).} To conclude the proof we derive an estimate for
$\|\Delta\tilde{\chi}_0\|_2$. Clearly,
\begin{displaymath}
  \rho_0^2\Delta\tilde{\chi}_0 =  \Div (\rho_0^2\nabla\tilde{\chi}_0) -
  2\rho_0\nabla\rho_0\cdot\nabla\tilde{\chi}_0 \,. 
\end{displaymath}
By (\ref{eq:33}c), (\ref{eq:141}), (\ref{eq:91}), and (\ref{eq:19}) we
obtain that, for some $\gamma>0$, 
\begin{equation}
\label{eq:144}
  |\Div (\rho_0^2\nabla\tilde{\chi}_0)|\leq \frac{C}{\epsilon}e^{-\frac{\gamma t}{\epsilon}} \,,
\end{equation}
and hence
\begin{displaymath}
  \|\Div (\rho_0^2\nabla\tilde{\chi}_0)\|_2\leq \frac{C}{\epsilon^{1/2}} \,. 
\end{displaymath}
As $\|\nabla\tilde{\chi_0}\|_2\leq2/\sqrt{3}$ for sufficiently small $\epsilon$, we
obtain that
\begin{displaymath}
  \|2\rho_0\nabla\rho_0\cdot\nabla\tilde{\chi}_0\|_2 \leq C\|\nabla\rho_0\|_2 \,.
\end{displaymath}
By Proposition \ref{prop:exist}, Lemma \ref{lem:outer-next-order},
(\ref{eq:141}), and (\ref{eq:91}) we obtain that, for some $\gamma>0$, 
\begin{equation}
\label{eq:145}
  |\nabla\rho_0|\leq C\Big(\frac{1}{\epsilon}e^{-\frac{\gamma t}{\epsilon}}+1\Big) \,. 
\end{equation}
and hence
\begin{equation}
\label{eq:108}
  \|\nabla\rho_0\|_2\leq\frac{C}{\epsilon^{1/2}} \,,
\end{equation}
yielding (\ref{eq:49}f).
\end{proof}

\begin{proof}[Proof of Theorem \ref{thm:stationary2}]
 Let  $\Hg$ be defined by 
\begin{displaymath}
   \Hg = \big\{ (\eta,\omega,\varphi)\in H^2(\Omega,\R^3)\,\big| \,
   (\nabla\eta,\nabla\omega,\nabla\varphi)\cdot{\bf n}\big|_{\partial\Omega}=0 \,;\, (\omega)_\Omega=0
    \big\} \,.
 \end{displaymath}
We equip $\Hg$ with the
norm
\begin{equation}
\label{eq:109}
   \|(\eta,\varphi,\omega)\|_\Hg = \|\eta\|_\infty + 
   \|\eta\|_{1,2} +  \|\varphi\|_{1,2}
+ \|\omega\|_{1,2}+ \epsilon( \|D^2\eta\|_2 +  \|D^2\varphi\|_2)+ \|D^2\omega\|_2 \,.
\end{equation}
Suppose that $(\eta,\omega,\varphi)\in\Hg$ is a solution of
\begin{subequations}
\label{eq:110}
\begin{empheq}[left={\empheqlbrace}]{alignat=2}
   &   -\Big(\Delta-\frac{1}{\epsilon^2}(3\rho_0^2-1+ |\nabla\tilde{\chi}_0|^2)\Big)\eta
    + \frac{2\rho_0}{\epsilon^2}\nabla\tilde{\chi}_0\cdot\nabla\omega  = f_1 & \quad
     \text{in } \Omega\,, \\
&-\Div(\rho_0^2\nabla\omega) - 2\Div \big(\eta\rho_0\nabla\tilde{\chi}_0\big) + \sigma_0\epsilon\Delta\varphi = f_2
& \quad
     \text{in } \Omega\,, \\
&   -\sigma_0\epsilon^2\Delta\varphi +\rho_0^2\varphi + 2\rho_0\tilde{\phi}_0\eta = f_3 & \quad
     \text{in } \Omega\,,
\end{empheq} 
\end{subequations}
where $(f_1,f_2,f_3)\in L^2(\Omega,\R^3)$. We split the remainder of the proof
into three steps.

\medskip\noindent{\em Step 1:} {\em Prove that $v=(\eta,\omega,\varphi)$ is well-defined.} To this end we
use the Lax-Milgram lemma. Let $w=(\tilde{\eta},\tilde{\omega},\tilde{\varphi})$,
and
\begin{displaymath}
  \Vg=\{ (\eta,\omega,\varphi)\in H^1(\Omega,\R^3)\,|\, (\omega)_\Omega=0\}\,. 
\end{displaymath}
Then define the bilinear form $B:\Vg\times\Vg\to\R$
\begin{multline}
\label{eq:111}
  B[v,w] = \langle\nabla\eta,\nabla\tilde{\eta}\rangle +
  \frac{1}{\epsilon^2}\Big\langle\Big(3\rho_0^2-1+|\nabla\tilde{\chi}_0|^2)\Big)\eta,\tilde{\eta}\Big\rangle + 
\frac{2}{\epsilon^2}\langle\rho_0\nabla\omega,\tilde{\eta}\nabla\tilde{\chi}_0\rangle +\\
 \frac{1}{\epsilon^2}\Big[  \langle\rho_0\nabla\omega,\rho_0\nabla\tilde{\omega}\rangle+2\langle\rho_0\nabla\tilde{\omega},\eta\nabla\tilde{\chi}_0\rangle 
  - \sigma_0\epsilon\langle\nabla\tilde{\omega},\nabla\varphi\rangle +\\ C_0\sigma_0\big( \sigma_0\epsilon^2\langle\nabla\varphi,\nabla\tilde{\varphi}\rangle +
  \langle\rho_0\tilde{\varphi},(\rho_0\varphi+2\eta\tilde{\phi}_0)\rangle\big)
\Big]\,,
\end{multline}
where $C_0>0$ is to be specified later. Since by (\ref{eq:97}) both $\tilde{\chi}_0$ and
$\tilde{\phi}_0$ are in $C^1(\Omega)$, it readily follows from the Sobolev
embeddings that there exists $C(\Omega,\epsilon)$ such that
\begin{displaymath}
  | B[v,w]| \leq C \|v\|_{1,2}\|w\|_{1,2} \,.
\end{displaymath}
By (\ref{eq:97}b) and Proposition \ref{lem:2.2} we have, for a sufficiently
small value of $\epsilon$,
\begin{displaymath}
\|1-\rho_0^2\|_\infty<\frac{1}{3} \,,
\end{displaymath}
and consequently, again for a sufficiently small value of $\epsilon$,
\begin{equation}
  \label{eq:112}
1<3\rho_0^2-1 < 3 \,.
\end{equation}

We next attempt to estimate $B(v,v)$ from below. Clearly,
\begin{multline*}
    B[v,v] = \|\nabla\eta\|_2^2 + \frac{1}{\epsilon^2}\|(3\rho_0^2-1)^{1/2}\eta\|_2^2
    +  \frac{1}{\epsilon^2}\Big[\|\eta\nabla\tilde{\chi}_0\|_2^2 +\|\rho_0\nabla\omega\|_2^2\\ +
    4\langle\rho_0\nabla\omega,\eta\nabla\tilde{\chi}_0\rangle - \sigma_0\epsilon\langle\nabla\omega,\nabla\varphi\rangle
    + C_0\sigma_0\big(\sigma_0\epsilon^2\|\nabla\varphi\|_2^2 +
    \langle\rho_0\varphi,(\rho_0\varphi+2\eta\tilde{\phi}_0)\rangle\big)\Big] \,.
\end{multline*}
 Let
\begin{displaymath}
  W=
  \begin{bmatrix}
    \epsilon^{-1}\eta \\
    \epsilon^{-1}\nabla\omega \\
    \sigma_0\nabla\varphi \\
    \epsilon^{-1}\varphi
  \end{bmatrix}\,,
\end{displaymath}
and
\begin{displaymath}
  M(x,j,\sigma_0,\epsilon,C_0)=
  \begin{bmatrix}
    3\rho_0^2 -1 + |\nabla\tilde{\chi}_0|^2 & 2\rho_0\nabla\tilde{\chi}_0 & 0 &
    C_0\sigma_0\rho_0\tilde{\phi}_0 \\
2\rho_0\nabla\tilde{\chi}_0 & \rho_0^2 & -\frac{1}{2} & 0 \\
0 & -\frac{1}{2} & C_0 & 0 \\
C_0\sigma_0\rho_0\tilde{\phi}_0 & 0 & 0 & C_0\sigma_0\rho_0^2
  \end{bmatrix}
\,.
\end{displaymath}
It can be easily verified that
\begin{displaymath}
   B[v,v] = \|\nabla\eta\|_2^2 + \int_\Omega W^tMW \,dx \,.
\end{displaymath}
Consequently, if we can show that for some $C_0>0$, $\inf_{x\in\Omega}\min \sigma(M)>\lambda_0>0$ (obviously
$\sigma(M)\subset\R$),  where $\lambda_0$ is independent of $\epsilon$,
coercivity of $B$ would follow. 

Let $j=\gamma j_R$ for some $j_R\in C^{3,\alpha}(\partial\Omega)$. For $\gamma \equiv0$ we have 
\begin{displaymath}
  M(x,0,\sigma_0,\epsilon,C_0)=
  \begin{bmatrix}
    2  & 0 & 0 &
    0 \\
0 & 1 & -\frac{1}{2} & 0 \\
0 & -\frac{1}{2} & C_0 & 0 \\
0& 0 & 0 & C_0 \sigma_0
  \end{bmatrix}
\end{displaymath}
Note that $\min \sigma(M(x,0,\sigma_0,\epsilon))>\lambda_0>0$ whenever $C_0>1/4$. We seek
$C_0(\gamma)$, satisfying $C_0(0)>1/4$, such that ${\rm
  det }\,M(x,0,\sigma_0,\epsilon,C_0(\gamma))>D_0>0$ for all $\gamma\in(0,\gamma_0)$. Here
$D_0$ is independent of $\epsilon$ and $\gamma_0$ is such that the solution of
\eqref{eq:15} with $j=\gamma_0j_R$ satisfies $\|\nabla\zeta\|_\infty^2<5/14-\sqrt{65}/70$. We note
that all elements of $M$ are uniformly bounded in $L^\infty(\Omega)$ as
$\epsilon\to0$, and hence $\sup_{x\in\Omega}\max \sigma(M)$ is also bounded.
Consequently, if such $C_0(\gamma)$ is found, it would follow that
$\inf_{x\in\Omega}\min \sigma(M)>\lambda_0>0$, for $\lambda_0$ which is independent of
$\epsilon$. 

We now write
\begin{displaymath}
  {\rm   det}\,M = C_0\sigma_0\rho_0^2\Big( -C_0^2\rho_0^2\sigma_0\tilde{\phi}_0^2
  +C_0\Big(\rho_0^2\big[3\rho_0^2 -1 -3|\nabla\tilde{\chi}_0|^2\big]
  +\frac{1}{4}\sigma_0\tilde{\phi}_0^2\Big)-\frac{1}{4}(3\rho_0^2 -1 +
  |\nabla\tilde{\chi}_0|^2)\Big) 
\end{displaymath}
Inside the boundary layer, where $d(x,\partial\Omega)<\delta$ , keeping $\sigma_0$ large
but fixed, we have  by (\ref{eq:97}b), (\ref{eq:102}),
(\ref{eq:107}), (\ref{eq:17.5}), and (\ref{eq:39})
\begin{displaymath}
  \rho_0\xrightarrow[\epsilon\to0]{}\rho_{i0}=\rho_r\mu\,.
\end{displaymath}
Similarly,  by (\ref{eq:106}), (\ref{eq:97}c), (\ref{eq:17.5}),
(\ref{eq:107}), (\ref{eq:chi-i0}), (\ref{eq:34}), (\ref{eq:39}), and
(\ref{eq:P2lt}) 
\begin{displaymath}
   |\nabla\tilde{\chi}_0|\xrightarrow[\epsilon\to0]{} P\,,
\end{displaymath}
 where $P$ is given in \eqref{eq:139}. 
Finally, by (\ref{eq:106}),  (\ref{eq:97}a), (\ref{eq:107}), and (\ref{eq:39})
\begin{displaymath}
   \sigma_0\tilde{\phi}_0^2\xrightarrow[\epsilon\to0]{}\rho_r^4\vartheta^2 \,.
\end{displaymath}
Consequently, as $\epsilon\to0$ we have
\begin{equation}
\label{eq:140}
  {\rm   det}\,M \to C_0\sigma_0\rho_r^2\mu^2 [aC_0^2+bC_0+c]\,,
\end{equation}
where $a$, $b$, and $c$ are given by \eqref{eq:133}. It follows by
Lemma \ref{lem:inner-solution-1-quadratic} that  there exists
$z_m<C_0<z_M$, such that  ${\rm   det}\,M >0$ for every
$\gamma\in(0,\gamma_0)$, where $z_m$ and $z_M$ are given by
\eqref{eq:129}. Outside the boundary layer, for $d(x,\partial\Omega)>\delta$, the
existence of such $z_m<C_0<z_M$ follows in precisely the same manner,
as we obtain \eqref{eq:140} once again with $a$, $b$ and $c$ given by
the limit $\eta\to\infty$ in \eqref{eq:133}.

From the foregoing discussion we may conclude that for some $z_m<C_0<z_M$,
$C(\Omega,\sigma_0)>0$, and sufficiently small $\epsilon$,
\begin{equation}
\label{eq:113}
   | B[v,v]| \geq
   \frac{C}{\epsilon^2}(\|\eta\|_2^2+\|\omega\|_{1,2}^2+\|\varphi\|_2^2+\epsilon^2\|\nabla\varphi\|_2^2) + \|\nabla\eta\|_2^2\,.
\end{equation}
We can thus
conclude the existence of a unique $v\in\Vg$ such that 
\begin{displaymath}
  B[v,w]=\langle F,w\rangle \quad \forall w\in\Vg \,,
\end{displaymath}
where  $F=(f_1,\epsilon^{-2}f_2,C_0\epsilon^{-2}f_3)$. Since $F\in L^2(\Omega,\R^3)$ it
follows by standard elliptic estimates that $v\in\Hg$ (note that
$\|v\|_\Hg\leq C_\epsilon\|v\|_{2,2}$). 

Let $(\rho_1,\chi_1,\phi_1)\in\Hg$. We set
\begin{subequations}
\label{eq:114}
  \begin{align}
  & f_1 = -h_1  - \frac{1}{\epsilon^2}\Big[|\nabla\chi_1|^2\rho_0  +
      \nabla\chi_1\cdot(2\nabla\tilde{\chi}_0+ \nabla\chi_1)\rho_1+(3\rho_0+\rho_1)\rho_1^2\Big]\\
  & f_2 =\epsilon \Div H_2 - \Div\big(\rho_1^2\nabla\tilde{\chi}_0\big)  - \Div
  \big(\rho_1(2\rho_0+\rho_1)\nabla\chi_1\big) \\
& f_3 = \epsilon^2h_3 + \rho_1(2\rho_0+\rho_1)\phi_1 + \rho_1^2\tilde{\phi}_0        \,.
\end{align}
\end{subequations}
Substituting the above into \eqref{eq:110} we can define the operator
$\A:\Hg\to \Hg$ by
\begin{displaymath}
  \A(\rho_1,\chi_1,\phi_1)=(\eta,\omega,\varphi)\,.
\end{displaymath}
We look for a fixed point of $\A$. 
 
\medskip\noindent{\em Step 2: Let $v=(\rho_1,\chi_1,\phi_1)\in\Hg$.  We prove that for
sufficiently small $\epsilon$ and $r(\epsilon)=
\epsilon^{5/4}$ it holds that}
\begin{equation}
 \label{eq:115}
   v\in B(0,r) \Rightarrow \A(v)\in B(0,r) \,.
 \end{equation}

To this end, let $0<r\leq\epsilon$ and $v\in B(0,r)$. We begin by obtaining a bound
on $\|\eta\|_2$ and $\|A(v)\|_{1,2}$. By (\ref{eq:114}a), 
\eqref{eq:112}, and the fact that $\|\nabla\tilde{\chi}_0\|_\infty\leq2/\sqrt{3}$
for sufficiently small $\epsilon$, we
have that
\begin{displaymath}
  \|f_1\|_{2} \leq \|h_1\|_{2} + \frac{C}{\epsilon^2}\big[\|\nabla\chi_1\|_{4}^2+
  \|\rho_1\|_{4}\big(\|\nabla\chi_1\|_{4}+\|\nabla\chi_1\|_{8}^2\big)+ \|\rho_1\|_4^2 + \|\rho_1\|_6^3 \big] \,.
\end{displaymath}
By (\ref{eq:49}a), Sobolev embeddings, and the fact that $\|\chi_1\|_{2,2}\leq r$,
we then obtain 
\begin{equation}
\label{eq:116}
    \|f_1\|_{2} \leq C_s\Big(\epsilon^{s/2} +\frac{r^2}{\epsilon^2}\Big)\,,
\end{equation}
where
\begin{displaymath}
  s=\iota - 4(1-\iota) \,.
\end{displaymath}

Similarly, we obtain, using (\ref{eq:49}f), that
\begin{multline*}
\|f_2\|_2 \leq C\big[\epsilon\|\Div H_2\|_{2}+ \epsilon^{-1/2}\|\rho_1\|_\infty^2   + 
\|\rho_1\|_\infty(\|\nabla\rho_1\|_2 +\epsilon^{-1/2}\|\nabla\chi_1\|_{1,2}) + \\ (\|\rho_1\|_\infty^2+\|\rho_1\|_\infty)\|\Delta\chi_1\|_{2} +
(1+\|\rho_1\|_\infty)\|\nabla\rho_1\|_4\|\nabla\chi_1\|_4 \Big] \,.
\end{multline*}
In the above, to obtain that 
\begin{equation}
\label{eq:142}
   \|\rho_1\nabla\rho_0\cdot\nabla\chi_1\|_2 \leq C\epsilon^{-1/2}\|\rho_1\|_\infty\|\nabla\chi_1\|_{1,2}\,,
\end{equation}
we used the following estimate
\begin{multline}
\label{eq:143}
  \|\rho_1\nabla\rho_0\cdot\nabla\chi_1\|_2\leq
  C\|\rho_1\|_\infty\Big(\|\nabla\rho_o\|_\infty\|\nabla\chi_1\|_2 \\
  + \Big\|\Upsilon\frac{\partial(\rho_i-\rho_a)}{\partial s}\Big\|_\infty\Big\|\frac{\partial\chi_1}{\partial s}\Big\|_{L^2(\Omega\setminus\Omega_{\delta})} +
  \Big\|\frac{\partial(\Upsilon[\rho_i-\rho_a])}{\partial t}\frac{\partial\chi_1}{\partial t}\Big\|_2\Big) \,.
\end{multline}
The last term on the right-hand-side of the above inequality can be
bounded as follows: we first write
\begin{multline*}
 \Big\|\frac{\partial(\Upsilon[\rho_i-\rho_a])}{\partial t}\frac{\partial\chi_1}{\partial t}\Big\|_2^2 \leq\\
 C\int_0^{|\partial\Omega|} \Big\|\frac{\partial(\Upsilon[\rho_i-\rho_a])}{\partial
   t}(s,\cdot)\Big\|_{L^2(0,\delta)}^2\Big\|\frac{\partial\chi_1}{\partial t}(s,\cdot)\Big\|_{L^\infty(0,\delta)}^2 \,ds \leq\\
 C\sup_{s\in(0,|\partial\Omega|)}\Big\|\frac{\partial(\Upsilon[\rho_i-\rho_a])}{\partial
   t}(s,\cdot)\Big\|_{L^2(0,\delta)}^2 \int_0^{|\partial\Omega|} \Big\|\frac{\partial\chi_1}{\partial t}(s,\cdot)\Big\|_{L^\infty(0,\delta)}^2 \,ds\,.
\end{multline*}
By (\ref{eq:141}) and (\ref{eq:91}) we have
\begin{displaymath}
   \sup_{s\in(0,|\partial\Omega|)}\Big\|\frac{\partial(\Upsilon[\rho_i-\rho_a])}{\partial
   t}(s,\cdot)\Big\|_{L^2(0,\delta)}^2 \leq \frac{C}{\epsilon}\,.
\end{displaymath}
Furthermore, as $\frac{\partial\chi_1}{\partial t}(s,\cdot)\in H^1_0(0,\delta)$, we may write
\begin{displaymath}
  \Big\|\frac{\partial\chi_1}{\partial t}(s,\cdot)\Big\|_{L^\infty(0,\delta)}^2\leq  \Big\|\frac{\partial\chi_1}{\partial t}(s,\cdot)\Big\|_{L^2(0,\delta)}\Big\|\frac{\partial^2\chi_1}{\partial t^2}(s,\cdot)\Big\|_{L^2(0,\delta)}
\end{displaymath}
Consequently, we have that
\begin{displaymath}
   \Big\|\frac{\partial(\Upsilon[\rho_i-\rho_a])}{\partial t}\frac{\partial\chi_1}{\partial t}\Big\|_2^2 \leq
   \frac{C}{\epsilon}\|\nabla\chi_1\|_{1,2}^2 \,.
\end{displaymath}
Substituting the above into \eqref{eq:143} yields \eqref{eq:142}.

Note, that since $\|v\|_{\Hg}\leq r$, we have by (\ref{eq:109}) that
$\|\chi_1\|_{2,2}\leq r$. To bound $\|\nabla\rho_1\|_4$ we use a standard interpolation
inequality \cite[Theorem 5.8]{adfo03} from
which we conclude that
\begin{displaymath}
  \|\nabla\rho_1\|_4 \leq C\|\nabla\rho_1\|_2^{1/2}\|\nabla\rho_1\|_{1,2}^{1/2} \leq C \frac{r}{\epsilon^{1/2}}\,.
\end{displaymath}
Hence,  
with the aid (\ref{eq:49}d) and Sobolev embeddings we find that
\begin{equation}
  \label{eq:117}
\|f_2\|_2 \leq C_s\Big(\epsilon^{1+s/2} +\frac{r^2}{\epsilon^{1/2}}\Big)\,.
\end{equation}
For later reference we need also the following estimate
\begin{multline*}
  |\langle\omega,f_2\rangle| \leq C\Big\{\epsilon\big[\|H_2\|_{L^2(\Omega\setminus\Omega_\delta)} \|\nabla\omega\|_{2}+ \|\Div
  H_2\|_{L^2(\Omega_\delta)} \|\omega\|_{2} +\|H_2\|_{L^\infty (\partial\Omega_\delta)} \|\omega\|_{L^1(\partial\Omega_\delta)}\big] \\+ \big[ \|\rho_1\|_4^2 +
  (\|\rho_1\|_\infty+\|\rho_1\|_\infty^2)\|\nabla\chi_1\|_2\big] \|\nabla\omega\|_{2} \Big\} \,,
\end{multline*}
and since
\begin{displaymath}
   \|\omega\|_{L^1(\partial\Omega_\delta)}\leq C\|\omega\|_{1,2}\,,
\end{displaymath}
we may conclude that
\begin{equation}
\label{eq:118}
|\langle\omega,f_2\rangle|\leq C_s(\epsilon^{2+s/2} +r^2)\|\omega\|_{1,2}\,.
\end{equation}

Finally, by (\ref{eq:114}c)
\begin{displaymath}
\|f_3\|_2 \leq C\{ \epsilon^2\|h_3\|_{2} + (\|\rho_1\|_{4} +
\|\rho_1\|_{8}^2)\|\phi_1\|_{4}+ \|\rho_1\|_{4}^2\big\} \,,
\end{displaymath}
which leads, by (\ref{eq:49}e), to a similar estimate
\begin{equation}
\label{eq:119}
 \|f_3\|_2 \leq C_s\Big(\epsilon^{2+s/2} + r^2\Big)\,.
\end{equation}
Combining the above  with \eqref{eq:116} and \eqref{eq:118} yields
\begin{equation}
\label{eq:120}
|\langle\A(v),F\rangle| \leq C_s\Big(\epsilon^{s/2} + \frac{r^2}{\epsilon^2}\Big)\big(\|\eta\|_{2} + \|\omega\|_{1,2} +\|\varphi\|_2\big)  \,. 
\end{equation}

As $B(\A(v),\A(v))=\langle\A(v),F\rangle$ we obtain by \eqref{eq:113} that
\begin{equation}
\label{eq:121}
 \|\eta\|_{2} + \|\omega\|_{1,2} +\|\varphi\|_2 \leq C_s[\epsilon^{2+s/2} +r^2] \,,
\end{equation}
and that
\begin{equation}
\label{eq:122}
\|\nabla\eta\|_{2}\leq \frac{C_s}{\epsilon} [\epsilon^{2+s/2}+r^2] \,.
\end{equation}

To complete the proof of \eqref{eq:115} we first rewrite
 (\ref{eq:110}c) in 
the form
\begin{displaymath}
    \sigma_0\epsilon^2\Delta\varphi = \rho_0^2\varphi + 2\rho_0\tilde{\phi}_0\eta - f_3 
     \text{ in } \Omega \,,
\end{displaymath}
from which we easily conclude, with the aid of \eqref{eq:121}
and \eqref{eq:119}, that
\begin{equation}
\label{eq:131}
  \|\Delta\varphi\|_2 \leq C_s \Big[\frac{r^2}{\epsilon^2} + \epsilon^{s/2}\Big]\,.
\end{equation}
As
\begin{displaymath}
  \|\nabla\varphi\|_2^2 =-\langle\varphi,\Delta\varphi\rangle \leq \|\varphi\|_2\|\Delta\varphi\|_2 \leq C_s \Big[\frac{r^4}{\epsilon^2} + \epsilon^{2+s}\Big]\,,
\end{displaymath}
we can conclude that
\begin{equation}
  \label{eq:130}
\|\nabla\varphi\|_2\leq C_s \Big[\frac{r^2}{\epsilon} + \epsilon^{1+s/2}\Big]
\end{equation}
By standard elliptic estimates we may conclude from \eqref{eq:131}
that 
\begin{equation}
  \label{eq:123}
\|D^2\varphi\|_2 \leq C_s(\Omega,\sigma_0,j_r) \Big[\frac{r^2}{\epsilon^2} + \epsilon^{s/2}\Big]\,.
\end{equation}

We next rewrite (\ref{eq:110}b) in the form 
\begin{displaymath}
  -\Div(\rho_0^2\nabla\omega) = 2\Div \big(\eta\rho_0\nabla\tilde{\chi}_0\big) -\sigma_0\epsilon\Delta\varphi + f_2
\end{displaymath}
As in the proof of \eqref{eq:123} we then obtain, with the aid of
(\ref{eq:117}), (\ref{eq:131}), \eqref{eq:144}, and \eqref{eq:145} that
\begin{displaymath}
  \|\Div(\rho_0^2\nabla\omega)\|_{2} \leq C_s \Big[\frac{r^2}{\epsilon} + \epsilon^{1+s/2}\Big]\,.
\end{displaymath}
Thus, by \eqref{eq:121} and (\ref{eq:97}b)
\begin{displaymath}
  \|\Delta\omega\|_{2} \leq \|\Div(\rho_0^2\nabla\omega)\|_{2} +
  \|\nabla\rho_0\|_\infty\|\nabla\omega\|_{2} \leq C_s \Big[\frac{r^2}{\epsilon} + \epsilon^{1+s/2}\Big]\,,
\end{displaymath}
from which, by standard elliptic estimates we obtain that
\begin{equation}
\label{eq:124}
\|D^2\omega\|_2\leq  C_s \Big[\frac{r^2}{\epsilon} + \epsilon^{1+s/2}\Big]\,.
\end{equation}
 
To complete the proof we need yet to bound $\epsilon\|D^2\eta\|_2$
 and $\|\eta\|_\infty$. To this end we
 rewrite (\ref{eq:110}a) in the form
 \begin{displaymath}
   -\Delta\eta = -\frac{1}{\epsilon^2}(3\rho_0^2-1+|\nabla\tilde{\chi}_0|^2)\eta
     - \frac{2}{\epsilon^2}\rho_0\nabla\tilde{\chi}_0\cdot\nabla\omega  + f_1 \,.
 \end{displaymath}
 It easily follows that
 \begin{equation}
\label{eq:125}
  \frac{1}{\epsilon^2}\|(3\rho_0^2-1+|\nabla\tilde{\chi}_0|^2)\eta\|_2 \leq  \frac{C}{\epsilon^2} (r^2
  + \epsilon^2) \,. 
 \end{equation}
 Furthermore, as
 \begin{displaymath}
   \|\rho_0\nabla\tilde{\chi}_0\cdot\nabla\omega\|_2 \leq Cr^2 \,,
 \end{displaymath}
 we obtain with the aid of \eqref{eq:125} and \eqref{eq:116} that
\begin{equation}
\label{eq:126}
  \|D^2\eta\|_2 \leq   C_s \Big[\frac{r^2}{\epsilon^2} + \epsilon^{s/2}\Big]\,,
\end{equation}
We can now employ Agmon's inequality (cf. \cite[Lemma 13.2]{ag65})
 from which we learn, using the above in conjunction with \eqref{eq:121}, that 
\begin{equation}
\label{eq:127}
  \|\eta\|_\infty \leq   \frac{C_s}{\epsilon} [r^2 + \epsilon^{2+s/2}] \,.
\end{equation}
Combining \eqref{eq:127} with \eqref{eq:121}, \eqref{eq:122},
\eqref{eq:126}, \eqref{eq:123}, and \eqref{eq:124} yields
\begin{displaymath}
  \|\A(v)\|_\Hg \leq  C_s\Big[\frac{r^2}{\epsilon} + \epsilon^{1+s/2}\Big]\,.
\end{displaymath}
We may thus choose $r=\epsilon^{5/4}$ to obtain that 
\begin{equation}
\label{eq:132}
   \|\A(v)\|_\Hg\leq C\epsilon^{1+s/2} <r 
\end{equation}
for
sufficiently small $\epsilon$ and $\delta$ and $1/2<s<1$.

\medskip\noindent{\em Step 3:  We prove that there exists
 a $\gamma<1$ such that  for all $(v_1,v_2)\in B(0,r)^2$ we have}
 \begin{equation}
   \label{eq:128}
 \|\A(v_1)-\A(v_2)\|_\Hg \leq \gamma\|v_1-v_2\|_\Hg \,.
 \end{equation}
 It can be easily verified that
 \begin{align*}
   & \|f_1(v_1)-f_1(v_2)\|_2 \leq \frac{C}{\epsilon^2}r\|v_1-v_2\|_\Hg \,, \\
   & \|f_2(v_1)-f_2(v_2)\|_2 \leq Cr\|v_1-v_2\|_\Hg \,,\\
 & \|f_2(v_1)-f_2(v_2)\|_2 \leq  Cr\|v_1-v_2\|_\Hg \,.
 \end{align*}
 Let now $\A(v_1)=(\eta_1,\omega_1,\varphi_1)$ and $\A(v_2)=(\eta_2,\omega_2,\varphi_2)$. As
 \begin{displaymath}
   B(\A(v_1)-\A(v_2),\A(v_1)-\A(v_2))=\langle\A(v_1)-\A(v_2),F(v_1)-F(v_2)\rangle \,, 
 \end{displaymath}
 we obtain  by \eqref{eq:113} that
 \begin{displaymath}
   \|\eta_1-\eta_2\|_2 \leq Cr\|v_1-v_2\|_\Hg \,.
 \end{displaymath}
 In the same manner used to derive \eqref{eq:121} and \eqref{eq:122} we
 can now obtain that
 \begin{displaymath}
   \|\omega_1-\omega_2\|_{1,2} + \|\varphi_1-\varphi_2\|_{1,2} \leq  Cr\|v_1-v_2\|_\Hg \,,
 \end{displaymath}
 and that
 \begin{displaymath}
     \|\nabla(\eta_1-\eta_2)\|_2 \leq \frac{Cr}{\epsilon}\|v_1-v_2\|_\Hg \,.
 \end{displaymath}
 We then proceed in precisely the same manner as in the derivation of
 \eqref{eq:123} and \eqref{eq:124} to obtain that 
 \begin{displaymath}
   \epsilon\|\omega_1-\omega_2\|_{2,2} +  \|\varphi_1-\varphi_2\|_{2,2} \leq  Cr\|v_1-v_2\|_\Hg \,.
 \end{displaymath}
 Finally, using the same procedure as in the derivation of
 \eqref{eq:126} and \eqref{eq:127} gives
 \begin{displaymath}
       \|\eta_1-\eta_2\|_\infty +\epsilon^2\|\eta_1-\eta_2\|_{2,2} \leq \frac{Cr}{\epsilon}\|v_1-v_2\|_\Hg \,.
 \end{displaymath}
 Combining all of the above then yields
 \begin{displaymath}
   \|\A(v_1)-\A(v_2)\|_\Hg\leq \frac{Cr}{\epsilon}\|v_1-v_2\|_\Hg\,,
 \end{displaymath}
 and since $r=\delta\epsilon$, we obtain \eqref{eq:128} for a sufficiently small
 value of $\delta$. 

We have thus established the existence in the ball $B((\rho_0,\chi_0,\phi_0),\epsilon^{1+s/2})$ in $\Hg$ of a unique solution for
\eqref{eq:7}. The
proof of \eqref{eq:11} follows immediately from \eqref{eq:121}. The
proof of \eqref{eq:12} similarly follows from \eqref{eq:109} and
\eqref{eq:132}. 
\end{proof}

\section*{Acknowledgments}
This work of all four authors was partially supported by the BSF grant
no. 2010194. YA was partially supported by the NSF Grant
DMS-1613471. The work of LB was partially supported by NSF grant
DMS-1405769. 
\bibliography{potential1}
 \end{document}